\documentclass[11pt,a4paper]{article}
\usepackage{fullpage,xspace,times}
\usepackage{amsmath,upgreek}
\usepackage{amsfonts}
\usepackage{amssymb}
\usepackage{epstopdf}
\usepackage{graphicx,epsfig}

\newtheorem{theorem}{Theorem}

\newtheorem{corollary}[theorem]{Corollary}

\newtheorem{lemma}[theorem]{Lemma}

\newcommand{\opt}{{\mbox{\textsc{opt}}}}
\newcommand{\pp}{{\mbox{{$\mathfrak{p}$}}}}
\newcommand{\lp}{{\mbox{{$\ell_{\pp}$}}}}

\newcommand{\pra}{\mbox{\textsc{Pr-makespan}}}
\newcommand{\prb}{\mbox{\textsc{Pr-SantaClaus}}}
\newcommand{\prc}{\mbox{\textsc{Pr-norm}}}

\newcommand{\B}{{\cal{B}}}

\newcommand{\J}{{\cal{J}}}

\newcommand{\Z}{{\mathbb{Z}}}

\newcommand{\eps}{{\upvarepsilon}}

\newcommand{\R}{{\cal{R}}}
\newcommand{\Xomit}[1]{ }
\newenvironment{proof}[1][Proof]{\textbf{#1.} }{\ \rule{0.5em}{0.5em}}

\mathchardef\mhyphen="2D

\begin{document}

\title{Efficient approximation schemes for scheduling \\ on a stochastic number of machines}
\date{}
\author{Leah
Epstein\thanks{ Department of Mathematics, University of Haifa,
Haifa, Israel. \texttt{lea@math.haifa.ac.il}. } \and Asaf
Levin\thanks{Faculty of Data and Decision Sciences, The Technion, Haifa, Israel. \texttt{levinas@ie.technion.ac.il}. Partially supported by ISF - Israel Science Foundation grant number 1467/22.}}

\maketitle

\begin{abstract}
We study three two-stage optimization problems with a similar structure and different objectives.  In the first stage of each problem, the goal is to assign input jobs of positive sizes to unsplittable bags. After this assignment is decided, the realization of the number of identical machines that will be available is revealed.  Then, in the second stage, the bags are assigned to machines.  The probability vector of the number of machines in the second stage is known to the algorithm  as part of the input before making the decisions of the first stage.  Thus, the vector of machine completion times is a random variable.  The goal of the first problem is to minimize the expected value of the makespan of the second stage schedule, while the goal of the second problem is to maximize the expected value of the minimum completion time of the machines in the second stage solution.
The goal of the third problem is to minimize the $\lp$ norm for a fixed $\pp>1$, where the norm is applied on machines' completion times vectors.
Each one of the first two problems admits a PTAS as Buchem et al. showed recently. Here we significantly improve all their results by designing an EPTAS for each one of these problems. We also design an EPTAS for $\lp$ norm minimization for any $\pp>1$.
\end{abstract}

{Keywords: Approximation algorithms; Approximation schemes; Two-stage stochastic optimization problems; Multiprocessor scheduling.}

\section{Introduction}
We consider scheduling problems where the goal is to  assign jobs non-preemptively to a set of identical machines.  Unlike traditional scheduling problems \cite{Gr66}, the number of identical machines available to the scheduler, denoted as $k$, is not a part of the input, but it is drawn from a known probability distribution on the set of integers $\{ 1,2,\ldots ,m\}$ for an integer $m\geq 2$ that is a part of the input, and it is given together with the probabilities. As a first stage, the decision maker completes an initial set of decisions, namely it assigns the jobs to $m$ bags, forming a partition of all input jobs.  Later, once the realization of the number of machines becomes known, it packs the bags to the machines, where the packing of a bag to a machine means that all its jobs are scheduled together to this machine. That is, in this later step, every pair of jobs that were assigned to a common bag will be assigned to a common machine, and the decision of the first stage (that is, the partition of input jobs into $m$ bags) is irrecoverable. This two-stage stochastic scheduling problem was recently introduced by Buchem et al. \cite{BEKRSW}.

Formally, the input consists of a set of jobs $\J=\{ 1,2,\ldots ,n\}$ where each job $j\in \J$ has a positive rational size $p_j$ associated with it.  We are given an integer $m\geq 2$ and let $\B =\{ 1,2,\ldots ,m\}$ denote the set of bags. We are given also a probability measure $q$ over the set of integers $\{ 1,2,\ldots, m\}$, where $q_k$ is a rational non-negative number denoting the probability that the number of identical machines in the resulting second stage instance is $k$.  Here, we will use the property that $\sum_{k=1}^m q_k=1$.   In the first stage of our problem, the jobs are split into $m$ bags by an algorithm.  Namely, a feasible solution of the first stage is a function $\sigma_1:\J \rightarrow \B$.  In the second stage, after the value of $k$ has been revealed, the bags are assigned to machines, such that all jobs of each bag are assigned together. Namely, the algorithm also computes assignment functions $\sigma_2^{k}$ defined for every realization $k$ in the support of $q$ of the bags to the set of integers $\{ 1,2,\ldots ,k\}$ denoting the indexes of machines in the instance of the second stage problem (corresponding to the realization of $q$ that is equal to $k$).  So if the realization of $q$ is $k$, a job $j\in \J$ will be assigned to the machine of index $\sigma_2^{(k)} (\sigma_1 (j))$.
The first function $\sigma_1$ maps jobs to bags, and the second function which is based on the value of $k$, $\sigma_2^{k}$, maps bags to machines. The number of bags is $m$ (and $\sigma_1$ is independent of $k$, so the partition into bags does not depend on $k$ which is not known yet at the time of assignment into bags), and the number of machines is $k$, so for every realization of $k$ the schedule of the jobs to the $k$ machines is a feasible (non-preemptive) schedule for $k$ identical machines.

We use the terminology of such scheduling problems and let $W_i$ be the work (also called load) of machine $i$ which is the total size of jobs assigned to machine $i$ in a schedule.  This is also the completion time of a unit speed machine processing continuously the set of jobs assigned to machine $i$ starting at time $0$ (in some order). The {\em makespan of the schedule in realization $k$} is the maximum work of a machine in the second stage solution in the realization $k$ of $q$, and the {\em Santa Claus value of the schedule of realization $k$} is the minimum work of a machine (among the $k$ machines) in the second stage solution in the realization $k$ of $q$.

The expected value of the makespan is the expected value of the random variable of the makespan of the schedule in realization $k$ (the expectation is computed based on all possible values of $k$).
The first problem that we consider is the problem of minimizing the expected value of this random variable.  We denote it as $\pra$ and refer to it as the {\em makespan minimization problem}.  The expected value of the Santa Claus value is the  expected value of the random variable of the Santa Claus value of the schedule in realization $k$. The second problem we consider is the problem of maximizing the expected value of this random variable.  We denote it as $\prb$ and refer to it as the {\em Santa Claus maximization problem}.   While the term makespan is used in the scheduling literature in this meaning the term Santa Claus is not traditional, and it is usually referred to as the minimum work of a machine in the schedule.
Given the length of the resulting terminology for our settings, we prefer to use the non-traditional terminology of Santa Claus value.
The expected value of the random variable is called cost or value for minimization problems, and it is called profit or value for maximization problems.
In the $\lp$ norm minimization problem, the objective for $k$ machines is $(\sum_{i=1}^k W_i^{\pp})^{1/\pp}$. The cost of a solution of the third problem is the expected cost based on the random variable. This last problem is denoted by \prc.

Since the stochastic nature of the problem results only from the different options for $k$ (and once $k$ is known, the problem is deterministic), we say that if the realization of $q$ is $k$, then this is the scenario of index $k$ or simply scenario $k$. We let $\opt_k$ denote the objective function value of the second stage problem in scenario $k$ for an optimal solution for the studied problem, and let $\opt=\sum_{k=1}^m q_k \cdot \opt_k$ denote the optimal value of an optimal solution to our problem.
An optimal solution needs to balance all scenarios, and we denote this optimal solution by $\opt$ as well (that is, we use the same notation as the cost or the value). We stress that $\opt_k$ is not necessarily the optimal objective function value for the input jobs and $k$ machines, since the solution with a fixed set of bags may be inferior.
We will assume that $\frac{1}{\eps}$ is a positive integer such that $\frac{1}{\eps} > 100$ (this can be done without loss of generality as we can always decrease the value of $\eps$ by a factor below $2$ to satisfy this condition). The assumptions on $\eps$ will be used for all schemes (and in particular we will use $\eps<\frac 1{100}$). Let $p_{\max}=\max_{j\in \J} p_j$.

Other models where the relations between schedules with different numbers of machines and a fixed set of jobs can be found in several articles (for example, see \cite{azar2000resource,brehob2000applying,RustogiS13}). See \cite{dye2003stochastic} for a discussion regarding the combination of two-stage stochastic problems and approximation algorithms.

\paragraph{An overview of our results.}
Here, we study both minimization problems and a maximization problem. All types of algorithms defined here are required to have polynomial running times (and our algorithms will in fact have strongly polynomial running times). For a minimization problem, an $\R$-approximation algorithm is an algorithm that always finds a solution that is feasible and it has a
cost at most $\R$ times the cost of an optimal solution.  For a maximization problem, an
$\R$-approximation algorithm is an
algorithm that always finds a feasible solution of
value at least $\frac{1}{\R}$ times the value of an optimal
solution.  The approximation ratio of a given algorithm is the infimum value of the parameter $\R$ such that this algorithm is an $\R$-approximation algorithm.

A polynomial time approximation scheme (PTAS) is a class of
approximation algorithms such that for any $\eps >0$ the class has a
$(1+\eps)$-approximation algorithm. An efficient
polynomial time approximation scheme (EPTAS) is a
stronger concept. This is a PTAS whose running time
is of the form $h(\eps) \cdot poly(n)$ where
$h$ is some (not necessarily polynomial but necessarily computable) function and $poly(n)$ is
a polynomial function of the length of the (binary) encoding of the input. We will justify the assumption $m\leq n$ later, and therefore the running time can be rewritten in this form when the polynomial is on $m$ and $n$.
A fully polynomial time approximation scheme (FPTAS) is an even
stronger concept, defined like an EPTAS, but the function $h$ must
be a polynomial in $\frac 1{\eps}$. In this paper, we are
interested in EPTAS's. Note
that a PTAS may have time complexity of the form
$n^{g(\eps)}$, where $g$ can be any function of $\frac 1{\eps}$ and even a power tower function. However, an EPTAS cannot have such a running time, which makes it more efficient.  The concept of an EPTAS is related to fixed parameter
tractable (FPT) algorithms (see \cite{DF99} for additional information on FPT algorithms).

In this paper, we focus on finding EPTAS's for the three objectives. We develop an EPTAS for the makespan minimization problem in Section \ref{sec:eptas1}.  Then, in Section \ref{sec:eptas2} we establish an EPTAS for the Santa Claus maximization problem \prb.  Last, in Section \ref{sec:eptas3} we modify our scheme for \pra\ in order to obtain an EPTAS for \prc. In the recent work \cite{BEKRSW}, a PTAS was designed for \pra, and another PTAS was designed for \prb. Those PTAS's are of forms that do not allow a simple modification of the PTAS into an EPTAS, and in particular for \prb\ a dynamic program over a state space whose size is exponential is used (a function of $\eps$ appears in the exponent of the input size), and for both schemes enumeration steps of such sizes are applied.

For all objectives, we assume that $n>m$ holds. If $n\leq m$, this means that every job can be assigned to its own bag. For each scenario, it is possible to apply an EPTAS for the corresponding problem (see for example \cite{alon1998approximation}, the details for previous work are discussed further below) and thereby get an EPTAS for the required problem.

The problems studied here generalize strongly NP-hard problems, and therefore one cannot expect to obtain an FPTAS (unless $P=NP$), and thus our results are the best possible. Specifically, the special case of each one of the three problems where the number of machines is known, the probability function has a single value $q_m$ that is equal to $1$, and other probabilities are equal to zero, is known to be NP-hard in the strong sense.
The three objectives studied here are the three main objectives previously studied for scheduling on identical machines.

\paragraph{Related work.}
Stein and Zhong \cite{stein2019scheduling} introduced the scheduling problem with an unknown number of machines. The problem was mostly studied with respect to makespan minimization.
In the deterministic variant \cite{stein2019scheduling}, the goal is to assign a set of jobs with known properties to identical machines. However, only an upper bound $m$ on the number of machines is given. Jobs have to be partitioned into $m$ subsets, such that every subset will act as an inseparable bag. Then, the number of machines $k$ (where $2\leq k \leq m$) becomes known, and the schedule is created using the bags, without unpacking them, as in the problem that we study here, and after the number of machines is revealed, the bags are assigned to the machines, as in the problem that we study.
Thus, this problem also has two steps or levels of computation, but the worst case out of all possible values of $k$ is analyzed, where the comparison for each value of $m$ is to an optimal offline solution for $k$ identical machines and arbitrary bags. It is not hard to see that a constant approximation ratio (of at most $2$) can be obtained using a round-robin approach even for jobs of arbitrary sizes (via pre-sorting). An improved upper bound of $\frac 53+\eps$ was shown using a much more careful assignment to bags \cite{stein2019scheduling}.

A variant where machines have speeds (similar to uniformly related machines) was defined and studied by \cite{eberle2023speed} with respect to makespan. In that version, the number of machines $m$ is known, but not their speeds. The number of required bags is equal to the number of machines, but some machines may have speed zero, so the case of identical machines and the makespan objective is seen as binary speeds in this work.
There are several input types of interest, which are arbitrary jobs, unit size jobs, sand (one large job that can be cut into any required parts), and pebbles (jobs that have arbitrary sizes, but they are relatively small) \cite{eberle2023speed,MS24}.  Tight bounds were proved for the case of sand and makespan with and without speeds \cite{stein2019scheduling,eberle2023speed}, and for the Santa Claus objective without speeds \cite{stein2019scheduling}. All these values are strictly smaller than $1.6$.
For sand, since any partition of the input jobs is allowed, linear programming can be used to find the best partition.
In the case with speeds  for arbitrary sizes of jobs the algorithm of \cite{eberle2023speed} has an approximation ratio of at most $2-\frac 1m$, while special cases allow smaller ratios  \cite{eberle2023speed,MS24}.

In \cite{BOSW23}, Balkanski et al. relate the problem of scheduling with an unknown number of machines and an unknown set of speeds to online algorithms with predictions. In online problems with predictions \cite{mitzenmacher2022algorithms}, the algorithm receives information on the input, where such information could have been computed via machine learning. This model takes into account both the situation where the input can be exactly according to the prediction, but also the option that the input is not according to the prediction. It is required for the algorithm to have a relatively good performance for every input, but the performance has to be better if the input was predicted correctly. In many algorithms the performance improves as the input gets closer to the predicted input. The work of \cite{BOSW23} provides results of this flavor.

As we mentioned, the most relevant work to our work is  \cite{BEKRSW}, where two PTAS's are provided. It is also mentioned in that work that FPTAS's for the case where $m$ is seen as a constant can be designed using standard methods. The problem is called {\it stochastic} due to the probability distribution on scenarios. However, since this distribution is simply a convex combination of the costs for different scenarios, the methods are related to those often used for deterministic algorithms and worst case analysis. The convex combination complicates the problem and it requires carefully tailored methods for the algorithmic design.
We follow the standard worst case studies of two-stage stochastic optimization problems, and we study the optimization problems of optimizing the expected value of the random variable.
However, optimizing the worst scenario is pessimistic and does not allow one to take information learned from previous data into account, while the study of the expected value allows us to prefer the typical scenarios in the bag assignment algorithm.

All three objectives were studied initially for identical machines, such that simple greedy approximation were presented first, some time later PTAS's were designed, and finally EPTAS's were designed or it was observed that one of the known PTAS's is an EPTAS or can be converted into an EPTAS, which is the best possible result for each one of the three cases unless $P=NP$.
We provide additional details for each one of the objectives.
The makespan objective was introduced by Graham \cite{Gr66,Gr69}, where greedy algorithms were studied (see also \cite{coffman1978application,friesen1984tighter}). Twenty years later a PTAS was designed \cite{HS87}. Hochbaum \cite{HocBook} mentions that one of the approaches of assigning large jobs of rounded sizes (namely, solving an integer linear program in fixed dimension) gives an EPTAS, and attributes this approach to D. Shmoys (see \cite{JKV20,ChenJZ18} for recent results).
The Santa Claus objective (where the name was introduced much later \cite{BansalS06}) was studied with respect to greedy heuristics
\cite{friesen1981analysis,Deuermeyer82,CKW92}. A PTAS which is actually an EPTAS was designed by Woeginger \cite{Woe97}. The $\lp$ norm objective function was studied with respect to greedy heuristics \cite{ChW75,CC76,LW95}, and a PTAS and an EPTAS were designed \cite{alon1997approximation}. The same authors generalized their results for other objectives that satisfy natural conditions \cite{alon1998approximation}. The $\lp$ norm of a vector is seen as a more suitable measure of fairness of a schedule compared to bottleneck measures  \cite{AAGKKV95,BP10}, and therefore we study this objective additionally to those studied in \cite{BEKRSW}.

\section{An EPTAS for the makespan minimization problem \label{sec:eptas1}}
We start with the makespan minimization problem, and continue to the other objectives in other sections. The first step will be to apply a discretization on bag sizes, and bound (from above) the set of relevant bag sizes as a function of $\opt$, where we enumerate possible values of $\opt$ and use rounding for this value as well. Our second step is to round the instance of jobs (which will be assigned to bags). We use standard guessing steps and rounding methods for these steps.
Then, we are going to approximate the rounded instance by using templates of the assignment of jobs to bags, and configurations of assigning templates to machines in every realization of $q$.  This will allow us to formulate an integer linear program that has a special structure, namely, it is a {\em two-stage stochastic integer linear program}.  This special structure, as well as trivial bounds on the parameters of this formulation, provides us with an FPT algorithm to solve the problem.  This algorithm has a running time dominated by a function of $\frac{1}{\eps}$ times a polynomial in $n$, and therefore it is an EPTAS.

We will be able to reduce the number of variables such that we can apply Lenstra's algorithm \cite{lenstra1983integer,Kan83} on the integer linear program, and the running time for solving this program will be smaller. This is the case also for the other two EPTAS's that we show, that is, we will apply Lenstra's algorithm for all objectives. The most difficult step in the current section is to reduce the number of different values of makespan for different scenarios. Without this step, the running time will not satisfy the requirements of an EPTAS. The reduction is based on linear grouping, which is typically used for bin packing \cite{FVL81} and not for scheduling. The other EPTAS's that we design require additional non-trivial steps, and we will discuss them later.

The optimal solution of this integer program is transformed into a set of decisions defining the functions $\sigma_1$ and $\sigma_2^{(k)}\ \forall k$.  Next, we present the details of the scheme together with its analysis. In what follows we will present steps where in each of those steps the resulting approximation ratio of the scheme increases by at most a multiplicative factor of $1+\Theta(\eps)$.  By scaling $\eps$ prior to applying the scheme, we get that the approximation ratio of the entire scheme will be at most $1+\eps$.   We will use the next lemma for our proof.

\begin{lemma}\label{boundss}
It holds that $\opt_k \in [p_{\max},n \cdot p_{\max}]$ for any $k \leq m$. Thus, $\opt \in [p_{\max},n \cdot p_{\max}]$.
\end{lemma}
\begin{proof}
In every scenario $k$, we have $\opt_k \geq p_{\max}$, as a job of size $p_{\max}$ needs to be assigned to a machine in every scenario.  Furthermore, $\opt_k \leq n \cdot p_{\max}$ because $\opt_k$ is total size of some job subset with at most $n$ jobs each of which of size at most $p_{\max}$.
Since $\opt$ is a convex combination of the values $\opt_k$, the property for $\opt$ is proved.
\end{proof}

\paragraph{Guessing $\opt$.}  Our first step is to guess the value of $\opt$ within a multiplicative factor of $1+\eps$.  That is, we would like to enumerate a set of polynomial number of candidate values that contains (at least one) value that is in the interval $[\opt, (1+\eps)\cdot \opt)$.  We will show that for a guess in this interval we indeed find a solution of cost $(1+\Theta(\eps))\cdot \opt$. The next lemma shows that it is possible to enumerate such a set of candidate values since $\log_{1+\eps} n \leq \frac n{\eps}$. Our algorithm performs this enumeration.

\begin{lemma}
There is a set consisting of $O(\log_{1+\eps} n)$ values such that this set of values has at least one value in the interval $[\opt, (1+\eps)\opt)$.
\end{lemma}
\begin{proof}
 We have that the number of integer powers of $1+\eps$ in the interval $[ p_{\max},n \cdot p_{\max}]$ is $O(\log_{1+\eps} n)$ and we can enumerate all of these values in linear time.  One of these values must be in the interval $[\opt, (1+\eps)\opt)$.
\end{proof}

In what follows, with a slight abuse of notation, we let $\opt$ be the value of this guessed candidate value (which is not smaller than $\opt$ and it is larger by at most a factor of $1+\eps$ if the guess is correct).  We will show that if there is a feasible solution to \pra\ whose expected value of the cost is at most $\opt$, then we will construct a feasible solution with expected value of the cost being at most $(1+\Theta(\eps))\cdot \opt$.  By the above lemma, this means that the problem admits an EPTAS as desired.

\paragraph{Rounding bag sizes.}
Now, we  provide a method to gain structure on the set of feasible solutions that we need to consider for finding a near optimal solution.  Given a feasible solution, the size of bag $i$ denoted as $P(i)$ is the total size of jobs assigned to this bag, that is, $P(i)=\sum_{j: \sigma_1(j)=i} p_j$.

\begin{lemma}\label{bigbag}
There exists a solution of cost at most $(1+3\eps)\cdot \opt$ in which the size of every bag is in the interval $[\eps \cdot \opt, (1+\eps) \cdot \opt]$.
\end{lemma}
\begin{proof}
We note that by our guessing, we have that in a feasible solution it holds that $P(i) \leq \opt$. This holds as every scenario has at least one machine with work of at least $P(i)$ (the machine that receives this bag), and thus $\opt_k \geq P(i)$ for every scenario $k$.
Let $S$ be the bag set in an optimal solution (whose cost is at most $\opt$) of sizes at most $\eps\cdot\opt$.  In order to prove the claim, we apply a bag merging process, and then modify the optimal solution based on the modification in the set of bags.  As long as there exists a pair of bags (of indexes) $i$ and $i'$ of sizes below $\eps\cdot \opt$, we unite such a pair of bags, by reassigning the jobs assigned to $i'$ to the bag $i$, that is, changing $\sigma_1$ so that every job $j$ for which it used to be that $\sigma_1(j)=i'$, its assignment is modified to $\sigma_1(j)=i$. The new set of bags resulting from $S$ is called $S'$. After applying the process of converting $S$ to $S'$, there might be at most one bag of size smaller than $\eps\cdot\opt$. This bag is kept separately and it is not a part of $S'$. Each bag of $S'$ has a total size of at least $\eps\cdot \opt$ and below $2\eps\cdot \opt < \opt$ since $\eps<\frac{1}{100}$.

After applying this change with respect to all such pairs (one pair at a time), in order to obtain a solution where $S$ was replaced with $S'$, the following modifications to $\sigma_2^{(k)}$ are applied for every value of $k$.
For a given scenario $k$, we apply the following process on every machine $i\leq k$.  We compute the total size $s(i)$ of the bags in $S$ assigned to $i$. At this time, the assignment is the one of the original optimal solutions, and we will modify it by replacing bags of $S$ with bags of $S'$ and modifying $\sigma_2^{(k)}$ accordingly. First, all bags of $S$ are removed from the machines (for scenario $k$). Then, the bags are assigned back to the machines, but this time the modified bags are used (that is, the bags of $S'$). Recall that the resulting size of any bag created by uniting the bags from $S$ into bags of $S'$ is at most $2\eps\cdot \opt$. Machine $i$ will get a subset of these bags of total size at most $s(i)+2\eps\cdot\opt$.
Bags are assigned such that one bag is added to the machine at each time, and this is done until the next bag in the list of bags in $S'$ does not fit the upper bound on the total size, in which case it has a total bag size of bags of $S'$ of at least $s(i)$.
Thus, since the bags of $S$ were assigned, all bags in $S'$ are assigned.  The bags not in $S$ are assigned exactly as they used to be assigned in $\sigma_2^{(k)}$ (they were not removed), so the resulting makespan in scenario $k$ is increased by at most $2\eps\cdot\opt$ with respect to $\opt_k$.

After this process is completed for all scenarios, consider the bag that was kept separately, if it exists. All the jobs of this bag are reassigned arbitrarily to one non-empty bag. This last step may increase the size of one bag by at most $\eps\cdot\opt$ and this increases the cost of the solution of each scenario by at most $\eps\cdot\opt$, so this bounds also the increase of the expected value of the cost of the resulting solution after the second stage. The last claim for the cost of an optimal solution holds by linearity of the expected value.
As a result of the modification, there are no bag sizes below $\eps\cdot\opt$, and since one bag may have a size that was increased by an additive term below $\eps\cdot\opt$, while previously no bag had a size above $\opt$, no bag has a size above $(1+\eps)\cdot\opt$.
\end{proof}

In what follows, we will increase the size of a bag to be the next value of the form $(\eps + r\eps^2)\cdot\opt$ for a non-negative integer $r$ (except for empty bags for which the allowed size remains zero). Thus, we will not use precise sizes for bags but upper bounds on sizes, and we call them {\it allowed sizes}.
We let $P'(i)$ be this increased (allowed) size of bag $i$, that is, $P'(i) = \min_{r: (\eps + r\eps^2)\opt \geq P(i)} (\eps + r\eps^2)\opt$.  Since for every subset of bags, the total allowed size of the bags in the set is at most $1+\eps$ times the total size of the bags in the subset, we conclude that this rounding of the bag sizes will increase the cost of any solution by a multiplicative factor of at most $1+\eps$ (and therefore the expected value also increases by at most this factor).
Thus, in what follows, we will consider only bag sizes that belong to the set $B=\{(\eps + r\eps^2)\opt : r=0,1,\ldots ,\frac{1}{\eps^2}\}\cup\{0\}$.  We use this set by Lemma \ref{bigbag}.
We conclude that the following holds.

\begin{corollary}
If the guessed value (\opt) is at least the optimal expected value of the makespan, then there is a solution of expected value of the makespan of at most $(1+\eps)(1+3\eps) \cdot \opt$ that uses only bags with allowed sizes in $B$.
\end{corollary}

Note that the allowed size of a bag may be slightly larger than the actual total size of jobs assigned to the bag. Later, the allowed size acts as an upper bound (without a lower bound), and we take the allowed size into account in the calculation of machine completion times in some cases (instead of the size).

\paragraph{Rounding job sizes.}
We apply a similar rounding method for job sizes.  Recall that by our guess, every job $j$ has size at most $\opt$ (see Lemma \ref{boundss}).  Next, we apply the following modification to the jobs of sizes at most $\eps^2 \opt$.  Whenever there is a pair of jobs, each of which has size at most $\eps^2 \opt$, we unite the two jobs (i.e., we delete the two jobs, adding a new job whose size is the sum of the two sizes of the two deleted jobs).  This is equivalent to restricting our attention to solutions of the first stage where we add the constraint that the two (deleted) jobs must be assigned to a common bag.  We repeat this process as long as there is such a pair.  If there is an additional job of size at most $\eps^2\opt$, we delete it from the instance, and in the resulting solution (after applying the steps below on the resulting instance) we add the deleted job to bag $1$. This is done after the entire algorithm completes running, so it may increase the makespan of every scenario and therefore the expected value of the makespan (by a multiplicative factor of $1+\eps^2$), but we do not take it into account in the algorithm or in calculating allowed sizes of bags.  This addition of the deleted job increases the makespan of every scenario by at most $\eps^2\opt$, so the resulting expected value of the makespan will be increased by at most $\eps^2\opt$.  We consider the instance without this possible deleted job, and we prove the following.
\begin{lemma}
The optimal expected value of the makespan of the instance resulting from the above modification of the job sizes among all solutions with allowed bag sizes in the following modified set
$B' = \{(\eps + r\eps^2)\opt : r=0,1,\ldots ,\frac{1}{\eps^2}+2\}\cup \{0\}$    is at most
$(1+2\eps) (1+\eps) (1+3\eps) \cdot \opt$.
\end{lemma}
\begin{proof}
We apply the following process on the solution $\sigma_1$ of the first stage.  For every non-empty bag $i$, we calculate the total size of jobs of sizes at most $\eps^2\opt$ assigned to $i$, and let $s_i$ be this value.  We delete the assignment of the jobs of size at most $\eps^2\opt$ from $\sigma_1$. We start defining the assignment of the new jobs that we created one by one in the following way.  We compute an upper bound $s_i+2\eps^2\opt$ on the total size of these new jobs that we allow to assign to bag $i$.  For every new job $j$, we assign it to a bag subject to the constraint that after this assignment, the new total size of new jobs in the bag will not exceed the upper bound.  This is possible, because whenever a bag $i$ cannot get a new job, the previously assigned new jobs to $i$ have total size at least $s_i$.  By summing over all $i$, and noting that before the modifications we started with a feasible solution (for the assignment of jobs to bags), we conclude that there is always a feasible bag to use.  In this process all bag sizes were increased by additive terms of at most $2\eps^2\opt$, and hence the modification of the allowed bag sizes.
The claim holds by noting that the total allowed size of each bag was increased by a multiplicative factor of at most $1+2\eps$ since previously all (non-zero) allowed bag sizes were at least $\eps \cdot \opt$.
\end{proof}

We next round up the size of every job $j$ to be the next value that is of the form $(1+\eps)^r \cdot \opt$ for an integer $r$. The size of every job cannot decrease and it may increase by a multiplicative factor of at most $1+\eps$. Job sizes are still larger $\eps^2 \cdot \opt$ and the sizes do not exceed $(1+\eps)\cdot \opt$. Thus, the number of different job sizes is $O(\log_{1+\eps} \frac {1+\eps}{\eps^2}) \leq \frac{2}{\eps^3}-1$.
We increase the allowed size of each non-empty bag by another multiplicative factor of $1+\eps$, and round it up again to the next value of the form $\eps + r\eps^2$.
Thus, the expected value of the makespan of a feasible solution increases by a multiplicative factor of at most $(1+\eps)^2$, where one factor of $1+\eps$ is due to the rounding of jobs, and the second one is due to bringing allowed sizes back to the form $\eps + r\eps^2$.
So by our guessing step, there is a feasible solution with expected value of the makespan of at most $(1+\eps)^3(1+2\eps)(1+3\eps) \opt$ that uses only bags with allowed size in  $B'' = \{(\eps + r\eps^2)\opt : r=0,1,\ldots ,\frac{1+2\eps}{\eps^2}+2\}\cup\{0\}$.

\medskip

A bag is called {\it tight} if the set of jobs of this bag cannot use a bag of a smaller allowed size.  Note that any solution where some bags are not tight can be converted into one where all bags are tight without increasing the cost. Note that the allowed size of a non-zero bag of allowed size $(\eps + r\eps^2)\opt$ can  be written in the form $(r+\frac{1}{\eps})\cdot \eps^2 \cdot \opt$, and since $\frac{1}{\eps}$ is integral, the allowed size of each bag is an integer multiple of ${\eps^2}\cdot \opt$.

\begin{lemma}\label{PandP1}
For any tight bag it holds that the allowed size of the bag is at most $\frac{1}{\eps}$ times its actual size (the total size of jobs assigned to the bag, such that their rounded sizes are considered).
\end{lemma}
\begin{proof}
For a tight bag with allowed size at least $(\eps+\eps^2)\cdot \opt$ the claim is obvious since the total size is larger than the allowed size minus $\eps^2\cdot \opt$.
For a tight bag with allowed size $\eps\cdot \opt$ this holds since an empty bag cannot be used for the contents of this bag (since it is tight), and no job has a rounded size below $\eps^2 \cdot \opt$. For tight bags of allowed size zero the property holds as well.
\end{proof}

\paragraph{Computing the maximum number of machines for which the assignment to bags does not matter.}
In our algorithm, we would like to focus on values of $k$ for which the structure of bags is crucial.  Now, we will detect values of $k$ for which there always exists a good assignment of bags to machines (for reasonable sets of bags). Let $P$ be the total size of all jobs (after the transformation applied on jobs for a fixed value of $\opt$). The total size (not the allowed size) of all bags is $P$, and we can compute the makespan based on the actual sizes of bags which are total sizes of jobs (after merging and rounding) assigned to the bags.
The motivation is that for very small values of $k$ the maximum bag size is so small compared to $\frac Pk$ that bags can be seen as very small jobs, and one can apply a greedy algorithm for assigning the bags, and still obtain a solution with a small makespan (based on bag sizes).

\begin{lemma}\label{eas}
If $k \leq \frac{ \eps \cdot P}{2\cdot\opt}$, then any set of bags (such that every job is assigned to a bag) where every bag has an allowed size (and size) not exceeding $2\opt$, leads to at least one schedule for $k$ machines with makespan in $[\frac{P}k,(1+\eps)\cdot \frac{P}k)$. For other values of $k$, the makespan of an optimal schedule using tight bags (based on their allowed sizes) is at most $({3}/{\eps^2})\cdot\opt$.
\end{lemma}
\begin{proof}
Consider a fixed value of $k$ and a fixed set of bags. We apply the so-called list scheduling algorithm that assigns one bag at a time. The makespan of the output is an upper bound on the optimal makespan in this scenario. Each bag is assigned to a machine such that the total size of the bags already assigned to that machine is minimized (breaking ties arbitrarily).  This can be applied for actual bag sizes or for allowed bag sizes.
Using the folklore analysis of list scheduling (see Graham \cite{Gr66} for the original proof), the makespan of the solution obtained by list scheduling on actual sizes is below $\frac{P}{k}+2\opt$ since the allowed size and the size of a bag is at most $2\opt$.
For $k \leq \frac{ \eps \cdot P}{2\cdot\opt}$, it holds that $2\opt \leq \eps \cdot \frac{P}k$ and since no solution (for rounded job sizes) exists in which the makespan is smaller than $\frac{P}{k}$, the claim holds.

Next, consider a value of $k$ such that $k > \frac{ \eps \cdot P}{2\cdot\opt}$.
Let $P'$ be the total allowed size of all bags, where $P'\leq \frac{P}{\eps}$ by Lemma \ref{PandP1}. By applying List Scheduling on bags with their allowed sizes, we have the following upper bound on the makespan:
$\frac{P'}{k}+2\opt \leq \frac{P}{\eps\cdot k}+2\opt  \leq 2 \opt\cdot (1+\frac 1{\eps^2})<3\cdot \frac{\opt}{\eps^2}$ by $\eps \leq \frac 1{100}$ and the lower bound on $k$.
\end{proof}

We use the property that no bag has an allowed size above $2 \cdot \opt$.
Our algorithm computes the maximum value of $k$ for which $2\opt \leq \eps \cdot \frac{P}{k}$ holds, and afterwards it excludes values of $k$ for which $2\opt \leq \eps \cdot \frac{P}{k}$ holds. We will have that the cost of the solution obtained for the scenario is at most $(1+\eps) \cdot \frac{P}{k}$, since bag sizes satisfy the condition for bags in the statement of Lemma \ref{eas}, and therefore we can use the first part of this lemma. This is a valid approach as adding a scenario where $2\opt \leq \eps \cdot \frac{P}{k}$ to the calculation of the approximation ratio will not increase the approximation ratio beyond the value $1+\eps$, and the running time for computing a good solution for such a value of $k$ is polynomial in $m,n$. Thus, we assume that $2\opt > \eps \cdot \frac{P}{k}$ holds for every $k$.

We consider the remaining scenarios. In what follows, we consider the assignment of jobs to bags and the assignment of the bags to machines in each scenario $k$ subject to the condition that $2\opt > \eps\cdot \frac{P}{k}$ and the optimal makespan of the scenario is at most $({3}/{\eps^2})\cdot\opt$.  In particular, it means that the number of bags of positive allowed sizes assigned to a machine in such a scenario is at most ${3}/{\eps^3}$, since we consider allowed bag sizes not smaller than $\eps\cdot \opt$ (if we exclude bags of allowed size zero).

\paragraph{Guessing an approximated histogram of the optimal makespan on all scenarios.}  Our next guessing step is motivated by linear grouping of the $\opt_k$ values.  To illustrate this step, consider a histogram corresponding to the costs of a fixed optimal solution satisfying all above structural claims, where the width of the bar corresponding to scenario $k$ is $q_k$ and its height is the value of the makespan in this scenario (that is, $\opt_k$).  Observe that when $k$ is increased, the optimal makespan does not  increase (without loss of generality, since the same assignment of bags can be used), so by plotting the bars in increasing order of $k$ we get a monotone non-decreasing histogram, where the total area of the bars is the expected value of the makespan of the optimal solution and the total width of all the bars is $1$ (since we assume that all scenarios are included in the histogram and every scenario $\kappa$ satisfies the condition $2\opt > \eps\cdot \frac{P}{\kappa}$ by scaling the probabilities of the remaining scenarios such that the total probability becomes $1$). We sort the indexes of scenarios with (strictly) positive probabilities ($q_k$ values) and we denote by $K$ the resulting sorted list.  For every $k\in K$, we compute the total width of the bars with indexes at most $k$ (using the ordering in $K$, i.e., in increasing indexes), and denote this sum by $Q_k$. We also let $Q'_k=Q_k-q_k$.  Thus, $Q_k$ is the total probability of scenarios in $\{1,2,\ldots,k\}$, and $Q'_k$ is the total probability of scenarios in $\{1,2,\ldots,k-1\}$. The set $K$ does not contain scenarios with probability zero, but if $q_{k-1},q_k>0$, it holds that $Q_{k-1}=Q'_k$. According to the definitions, it holds that $Q_k-Q'_k=q_k$ for $k\in K$, and the interval $[Q'_k,Q_k)$ on the horizontal axis corresponds to scenario $k$.

Next, we compute an upper bound histogram $W$ as follows. We start with selecting a sublist $K'$ of $K$. The motivation is that schedules will be computed later only for scenarios in $K'$, and they will be copied to scenarios of some of the larger indexes.
The set $K'$ acts as a representative set, and we show that it is possible to restrict ourselves to such a set. Since we increase upper bounds on the makespan for some scenarios, the feasibility is not harmed.

An index $k\in K$ belongs to $K'$ if there exists an integer $\ell$ such that $Q'_{k}\leq \eps^3\ell < Q_{k}$. The motivation is that we would like to consider points which are integral multiples of $\eps^3$ and the set $K'$ contains all values of $k \in K$ for which such special values belong to the interval $[Q'_k,Q_k)$. Values of makespan will be increased such that the makespan function will still be piecewise linear, but it will have a smaller number of image values.

For every $k \in K'$, the new upper bound histogram is defined as $\opt_k$ for all points with horizontal value between $Q'_k$ and up to $Q'_{k'}$ where $k'>k$ is the index just after $k$ in the sublist $K'$ or up to the last point where the histogram is defined ($Q_t$ for the maximum value $t\in K$) if $k$ is the largest element of $K'$.  Since the original histogram was monotone non-increasing, the new histogram is pointwise not smaller than the original histogram. The possible modification is in the interval $[Q_k,Q'_{k'})$ if $k$ is not the maximum value of $K'$, and in this case there may be a change for $k+1,\ldots,k'-1$ (that is, if $k'\geq k+2$). If $k$ is the largest element of $K'$ but not of $K$, there may be a change for all $t\in K$ such that $t\geq k+1$.
Thus, an upper bound on the total area below the histogram $W$ is not smaller than the expected value of the makespan of the optimal solution.

We use the modified histogram $W$ to obtain another bound. For that we see $W$ as a step function whose values are the corresponding points in the upper edge of the histogram.  We define a new histogram by letting it be $W(x+\eps^3)$ for all $x\in [0,1]$ (and zero for cases that $W$ is undefined due to an argument above $1$).  In this way we delete a segment of length $\eps^3$ from $W$ and we shift the resulting histogram to the left. Since the makespan for every scenario never exceeds $({3}/{\eps^2})\cdot\opt$, every point in $W$ has a height of at most $\frac{3\opt}{\eps^2}$ and we deleted a segment of width of $\eps^3$, the total area that we delete is at most $3\eps\opt$.  The resulting histogram is pointwise at most the original histogram (and thus also not larger than $W$).  This property holds since every value $\opt_k$ was extended to the right for an interval not longer than $\eps^3$. Thus, if we consider $W$ instead of the original histogram, we increase the cost of the solution by a multiplicative factor not larger than $(1+3\eps)$.

The guessing step that we use is motivated by this function $W$.  We let $W_k$ be the height of the histogram $W$ in the bar corresponding to the scenario $k$.  The relevant values of $k$ are those in $K'$, and the histogram $W$ only has values of the form $\opt_k$ for $k \in K'$. We guess the histogram $W$.  This means to guess the optimal makespan of $O(\frac{1}{\eps^3})$ scenarios, where makespans can be one of at most $\frac{3}{\eps^4}$ different values. This holds since all allowed bag sizes are integer multiples of $\eps^2\cdot \opt$, so makespans are also integer multiples of $\eps^2\cdot \opt$, and the makespan of each scenario is at most $\frac{3\opt}{\eps^2}$.  Thus, the number of possibilities of this guessing step is upper bounded by a function of $\frac{1}{\eps}$ which is $O({{(\frac 1{\eps})}^{O(1/\eps^3)}})$.

In what follows, we consider the iteration of the algorithm defined below when we use the value of the guess corresponding to the optimal solution.  Algorithmically, we will try all possibilities, check the subset of those for which we can provide a feasible solution, and output the best feasible solution obtained in this exhaustive enumeration. The running time of the next algorithm is multiplied by the number of possibilities.

\paragraph{The template-configuration integer program.}  A {\em template} of a bag is a multiset of job sizes assigned to a common bag.  We consider only multisets for which the total size of jobs is at most $(1+6\eps) \cdot \opt$ since the largest allowed size of any bag is smaller.  Since the number of distinct job sizes (in the rounded instance) is upper bounded by $\frac{2}{\eps^3}-1$ and there are at most $\frac{2}{\eps^2}$ jobs assigned to a common bag (since the rounded size of each job is above $\eps^2\cdot \opt$), we conclude that the number of templates is at most $(\frac{2}{\eps^3})^{(2/\eps^2)}$, which is a constant depending only on $\frac{1}{\eps}$.
The reason is that each bag has $\frac{2}{\eps^2}$ slots for jobs, such that each one may be empty or contain a job of one of the sizes. This calculation proves an upper bound on the number of possible templates, and we use a single template for every multiset of job sizes even when one multiset can be found in more than one way in the last calculation.

We are going to have a non-negative counter decision variable $y_t$ for every template $t$, where this variable stands for the number of bags with template $t$.  Let $\tau$ be the set of templates, and assume that a template is represented by a vector. The length of such a vector is the number of different (rounded) job sizes, and for every index $\ell$, the $\ell$-th component of the vector of the template (denoted by $t_{\ell}$) is the number of jobs with size equal to the $\ell$-th job size that are packed into a bag with this template.  In order for such an assignment of the first stage to be feasible, we have the following constraints where $n_{\ell}$ denotes the number of jobs in the rounded instance whose size is the $\ell$-th job size of the (rounded) instance.

$$ \sum_{t\in \tau} y_t \leq m \ , $$
$$ \sum_{t\in \tau} t_{\ell} \cdot y_t = n_{\ell} \  , \forall \ell .$$

The first constraint states that the number of bags is at most $m$. If this number is strictly smaller than $m$, then the remaining bags have allowed sizes of zero and they will not contain jobs.
The second constraint, which is a family of constraints, states that the correct numbers of jobs of each size are assigned to templates, according to the number of copies of each template.
Observe that the number of constraints in this family of constraints is a constant depending on $\eps$ (it is at most $\frac{2}{\eps^3}$), and all coefficients in the constraint matrix are (non-negative) integers not larger than $\frac{2}{\eps^2}$.

We augment $K'$ with the minimum index in $K$ if this index does not already belong to $K'$, in order to satisfy the condition that for every index of $K$ there is some index of $K'$ that is not larger.
Consider a scenario $\kappa \in K'$ (satisfying the condition that $2\opt > \eps\cdot \frac{P}{\kappa}$), and the optimal makespan of scenario $\kappa$, which is at most $\frac{3\opt}{\eps^2}$.  Recall that in scenario $\kappa$  the number of non-empty bags assigned to each machine is at most  $\frac{3}{\eps^3}$. Define a configuration of a machine in scenario $\kappa$ to be a multiset of templates such that the multiset has at most $\frac{3}{\eps^3}$ templates (counting multiple copies of templates according to their multiplicities) and the total size of the templates encoded in the multiset is at most $W_{\kappa}$, which is the guess of a value of the histogram $W$. The number of configurations is at most $((\frac{2}{\eps^3})^{(2/\eps^2)})^{\frac{3}{\eps^3}}$, and this is also an upper bound on the number of suitable configurations (whose total allowed size of bags does not exceed $W_{\kappa}$). Since our mathematical program will not limit the makespan of any scenario, we use an upper bound for it by not allowing configurations whose total allowed sizes exceed the planned makespan for the scenario. We let $C^{(\kappa)}$ denote the set of configurations for scenario $\kappa$, where $c\in C^{(\kappa)}$ is a vector of $|\tau|$ components where component $c_t$ for $t\in \tau$ is the number of copies of template $t$ assigned to configuration $c$.  Components are non-negative integers not larger than $\frac{3}{\eps^3}$.
For scenario $\kappa \in K'$, we will have a family of non-negative decision variables $x_{c,\kappa}$ (for all $c\in C^{(\kappa)}$) counting the number of machines with this configuration.

For each such scenario $\kappa$, we have a set of constraints each of which involves only the template counters (acting as a family of global decision variables) and the family of the $\kappa$-th local family of decision variables, namely the ones corresponding to this scenario.  The family of constraints for the scenario $\kappa \in K'$ are as follows.

$$ \sum_{c\in C^{(\kappa)}} x_{c,\kappa}= \kappa \ ,$$
$$ \sum_{c\in C^{(\kappa)}} c_t\cdot x_{c,\kappa} - y_t = 0 \ , \ \forall t\in \tau \ .$$

Observe that the number of constraints in such a family of constraints is upper bounded by a constant depending only on $\eps$ (which is the number of possible templates plus $1$) and the coefficients in the constraint matrix are again all integers of absolute value at most a constant depending only on $\eps$ (at most $\frac{3}{\eps^3}$.).

All decision variables are forced to be non-negative integers and we would like to solve the feasibility integer program defined above (with all constraints and decision variables).
The right hand side vector consists of integers that are at most $n$ (using $m\leq n$).
Such an integer linear program is a two-stage stochastic IP. Using the property that the number of scenarios in $K'$ is upper bounded by a function of $\frac 1{\eps}$, we conclude that the integer linear program has a fixed dimension. Thus,
we use Lenstra's algorithm to solve it instead of an algorithm for two-stage stochastic IP.

We obtain a feasible solution $(x,y)$ if such a solution exists.  Based on such a feasible solution, we assign jobs to bags using the $y$ variables.  That is, for every $t\in \tau$, we schedule $y_t$ bags using template $t$.  By the constraint $ \sum_{t\in \tau} y_t \leq m$ there are at most $m$ bags, and the other bags will be empty.  Next, for every bag and every size $p$ of jobs in the rounded instance, if the template assigned to the bag has $\alpha$ jobs of this size, we will assign $\alpha$ jobs of this size to this bag.  Doing this for all sizes and all bags is possible by the constraints $ \sum_{t\in \tau} t_{\ell} \cdot y_t = n_{\ell} \  , \forall \ell$, and these constraints ensure that all jobs are assigned to bags.  Note that the modification for jobs consisted of merging small jobs. When such a merged job is assigned, this means that a subset of jobs is assigned. Reverting jobs to their original sizes does not increase any of the costs, since no rounding steps decreased any sizes. Next consider the assignment of bags to machines in each scenario.  Consider a scenario $\kappa' \in K$, and let $\kappa$ be the largest index in $K'$ such that  $\kappa\leq \kappa'$. Assign $x_{c,\kappa}$ machines to configuration $c$ (for all $c\in C^{(\kappa)}$).  It is possible by the constraint  $ \sum_{c\in C^{(\kappa)}} x_{c,\kappa}= \kappa \leq \kappa'$.  If a configuration $c$ assigned to machine $i$ is supposed to pack $c_t$ copies of template $t$, we pick a subset of $c_t$ bags whose assigned template is $t$ and assign these bags to machine $i$.  We do this for all templates and all machines.  In this way we assign all bags by the constraints $ \sum_{c\in C^{(\kappa)}} c_t\cdot x_{c,\kappa} - y_t =0 \ , \ \forall t\in \tau$.
If $\kappa'>\kappa$, at least one machine will not receive any configuration and therefore it will not receive any bags or jobs, and we refer to such a machine as empty.
Since the configurations we used in scenario $\kappa$ have a total size of jobs of at most $W_{\kappa}$, we conclude the following.

\begin{corollary}
For every value of the guessed information for which the integer program has a feasible solution, there is a linear time algorithm that transform the feasible solution to the integer program into a solution to the rounded instance of \pra\ of cost at most $\sum_{q} q_{\kappa} \cdot W_{\kappa}$.
\end{corollary}

Our scheme is established by noting that an optimal solution satisfying the assumptions of the guessed information provides us a multiset of templates all of which are considered in $\tau$ and a multiset of configurations (and all of them have total size of templates not larger than $W$) for which the corresponding counters satisfy the constraints of the integer program.  Thus, we conclude our first main result.

\begin{theorem}
Problem \pra\ admits an EPTAS.
\end{theorem}

\section{An EPTAS for the Santa Claus problem \label{sec:eptas2}}
In this section we apply additional ideas to obtain an EPTAS for the second problem.
We present the details of the scheme together with its analysis. In what follows, and similarly to the scheme to \pra, we will present steps, and in each of those steps the resulting approximation ratio of the scheme increases by at most a multiplicative factor of $1+\Theta(\eps)$.

\subsection{Preprocessing steps}
We consider an optimal solution $\opt$, and analyze the information that one can guess in order to be able to approximate it.
We assume without loss of generality that $\opt_k\leq \opt_{k'}$ for all $k>k'$, since a solution for scenario $k$ can be used for scenario $k'$ (by assigning the bags of $k-k'$ machines in an arbitrary way).
Recall that we can assume that for every value of $k$ in the support of $q$ the instance has at least $k$ non-zero sized jobs, since we assume $n>m$. We will use the next lemma for our proof.

\begin{lemma}\label{boundssanta}
For any scenario $k$ (such that $1 \leq k \leq m$) and an assignment of jobs to bags according to the solution $\opt$, there exists a job $j^k$ for which it holds that $p_{j^k} \leq \opt_k \leq n \cdot p_{j^k}$.
\end{lemma}
\begin{proof}
Consider the given assignment of jobs to bags and an optimal assignment of the bags to $k$ machines. If there is an empty bag, then we can move an arbitrary job from another bag that used to contain at least two jobs, to this empty bag without decreasing the value of \opt\ in every scenario (simply by assigning the two corresponding bags to a common machine in each scenario). Since for every scenario $k$, we have $m\geq k$, we conclude that in every scenario $\opt_k >0$. Consider a machine of minimum load, and let $j^k$ be a job of maximum size assigned to this machine.  As observed above such a job exists. By the choice of the machine, it holds that $\opt_k \geq p_{j^k}$. On the other hand, this machine has at most $n$ jobs, and none of the jobs is larger than $p_{j^k}$. The last observation implies that $\opt_k \leq n \cdot p_{j^k}$.
\end{proof}

We would like to split the sequence of scenarios into four parts (where some of the parts may be empty). The suffix (with the largest numbers of machines) will contain scenarios for which the assignment of bags is unimportant since the gain from them in the objective function value $\opt$ is small either because the probability is zero or because the number of machines is large. This suffix may be empty. There will also be a prefix (which could be empty) of machines for which the specific assignment to bags is unimportant because the number of machines is small and any reasonable assignment into bags allows a good schedule.  For this prefix we will use different arguments from the ones we used for \pra. For the remaining scenarios, the prefix will consist of scenarios for which the gain is also small, and a suffix of the most important scenarios.

The first step is to guess the maximum scenario $k_{\max}$ for which $\opt_{k_{\max}}\geq \eps \cdot \opt$ holds, out of scenarios with strictly positive probabilities, that is, such that $q_{k_{\max}}>0$ holds.
This index is well-defined since there exists an index $k$ for which $q_k>0$ and $\opt_k \geq \opt$.
There are at most $m$ possible values for this guess.
 While $\opt$ still denotes an optimal solution, we do not guess or use its value as a part of the algorithm. Let $LB$ be equal to $\opt_{k_{\max}}$ rounded down to the next integer power of $1+\eps$. For a fixed job $j^{k_{\max}}$ (see Lemma \ref{boundssanta}), the number of possible values for $LB$ is $O(\frac{n}{\eps})$. Since the index $j^{k_{\max}}$ is also not known, the number of possible values for $LB$ is $O(\frac{n^2}{\eps})$.

The proof of the next lemma is obtained by selecting a set of consecutive scenarios minimizing the total weighted profit out of $\frac{1}{\eps}$ disjoint sequences of scenarios of similar forms. The proof of the lemma requires the knowledge not only of $\opt$ but of all values of the form $\opt_i$. However, we use the lemma to obtain the information that given the correct value $LB$ (this is a rounded value, so we will be able to guess it), there is a pair of values $k,k'$ and a value $\rho$ that specify the only properties of $\opt$ that we use for our algorithm.

\begin{lemma}\label{yy}
There is a value of $\rho$ which is an integer power of $\frac{1}{\eps}$ that satisfies $\frac{1}{\eps^2} \leq \rho \leq (\frac{1}{\eps})^{1/\eps+1}$ and there exist two indexes $k,k'$ such that $0 \leq k'<k \leq k_{\max}$, where every non-zero index has to be a scenario in the support of $q$, such that the following conditions hold.
\begin{enumerate}
\item $\opt_k \leq \rho \cdot LB$,
\item if $k'\geq 1$, then $\opt_{k'} \geq \frac{\rho}{\eps} \cdot LB$,
\item $\sum_{\kappa=k'+1}^{k-1} q_{\kappa} \cdot \opt_{\kappa} \leq \eps\cdot\opt$.
\end{enumerate}
\end{lemma}
\begin{proof}
For every value of $\rho$ that is an integer power of $\frac{1}{\eps}$ in the interval given in the statement of the lemma, that is, for $\rho=\frac{1}{\eps^r}$ for $r=2,3,\ldots ,\frac{1}{\eps}+1$, we compute  the total (weighted) profit $L_r$ of scenarios for which the objective value for the solution for this scenario in the optimal solution is within the interval $[LB \cdot \frac{1}{\eps^r},LB \cdot \frac{1}{\eps^{r+1}})$. That is, the value $L_r=\sum_{{\kappa}: LB \cdot \frac{1}{\eps^r} \leq \opt_{\kappa} < LB \cdot \frac{1}{\eps^{r+1}}} q_{\kappa}\cdot \opt_{\kappa}$ is computed for all values of $r$.

Since the values $\opt_{\kappa}$ form a monotonically non-increasing sequence, the scenarios contributing to $L_r$ for a fixed value of $r$ form a consecutive subsequence of scenario indexes, the scenarios for all values of $r$ also form a consecutive subsequence of scenario indexes as those are scenarios ${\kappa}$ with $\opt_{\kappa} \in [LB \cdot \frac{1}{\eps^2},LB \cdot \frac{1}{\eps^{1/\eps+2}})$, and the different values of $r$ form a partition of the last subsequence.

Let $r'$ be a value of $r$ for which $L_r$ is minimized, and define $\rho=\frac{1}{\eps^{r'}}$.  Let $k$ be the minimum index of a scenario in the support of $q$ for which $\opt_k \leq \rho \cdot LB$ and $k'$ be the maximum index of a scenario for which $\opt_{k'} \geq \frac{\rho}{\eps} \cdot LB$. Since $\opt_{k_{\max}} \leq LB \cdot (1+\eps) < \rho\cdot LB$, the integer $k$ is well-defined. If $k'$ does not exist in the sense that there is no scenario $k''$ for which $\opt_{k''} \geq \frac{\rho}{\eps} \cdot LB$, we let $k'=0$. It holds that $k>k'$, but it is possible that $k=k'+1$, in which case the sum of the third condition is empty.

  The sum of $L_r$ for different values of $r$ is taken over a weighted contribution of $\frac 1{\eps}$ disjoint subsets of  scenarios to $\opt$, and thus the values of $L_r$ have a sum not exceeding $\opt$. A minimizer of $L_r$ satisfies the requirement by averaging, and thus the third condition holds by our choice of $r$.
 \end{proof}

\paragraph{The next guessing step.} In this step we guess several additional values. As mentioned above, we guess the value of $LB$ (which is a power of $1+\eps$ and it is equal to $\opt_{k_{\max}}$ within a multiplicative factor of $1+\eps$ since $\opt_{k_{\max}} \in [LB,LB\cdot (1+\eps) )$), the value of $\rho$ (by guessing $r'$ from the proof of Lemma \ref{yy}, that is, we guess the power of $\frac{1}{\eps}$), and two indexes $k'<k\leq k_{\max}$.  That is, we would like to enumerate a set of polynomial number of candidate values of four components vectors that contains at least one vector with first component value that is $LB$, with a second component whose  value is $\rho$, and the two indexes $k'<k$ in the third and fourth components of that vector. The next lemma shows that it is possible to enumerate such a set of candidate values. Our algorithm performs this enumeration.

\begin{lemma}
There is a set consisting of $O((\frac{m\cdot n}{\eps})^2)$ vectors such that this set of vectors contains at least one vector with the required properties. Thus, the number of guesses including the guess of $k_{\max}$ is at most $O((\frac{n}{\eps})^2\cdot m^3)$.
\end{lemma}
\begin{proof}
The bound of $m$ on the number of possibilities for the values of $k$ and of $k'$ is based on the number of scenarios. The number of options for $\rho$ is $\frac{1}{\eps}$, and the number of options for the rounded value of $LB$ is $O(\frac{n^2}{\eps})$. The earlier guess of $k_{\max}$ changes the power of $m$ from $2$ to $3$.
\end{proof}

In what follows, we let $LB$ be the first component value of the guessed information, $\rho$ be the second component value of the guessed information, and $k,k'$ be the two guessed scenarios (where it is possible that $k'=0$ is not an actual scenario).  We will show that if there is a feasible solution to \prb\ for which the guessed information describes the solution and its expected value of the objective is \opt, then we will construct a feasible solution with expected value of the objective being at least $(1-\Theta(\eps))\cdot \opt$.  This means that the problem admits an EPTAS as desired.

We let $UB=LB \cdot \rho$.  Furthermore, for every scenario $\kappa$ with $k'<\kappa< k$, we set $q_{\kappa}=0$. From now on we drop the assumption that the sum of all $q$ values is $1$ and instead of that we use the assumption that these are non-negative numbers with sum at most $1$. This modification of the vector ($q$) decreases the expected value of the Santa Claus value of the optimal solution by at most $\eps\cdot\opt$ (and it does not increase the objective function value of any feasible solution to our problem).
The justification for this is as follows. From Lemma \ref{yy} it follows that there is a quadruple of integers $k,k',\rho,\log_{1+\eps}LB$ (where each integer belongs to the set that we test for this integer) for which there is a solution of profit at least $\frac{\opt}{1+\eps} - \eps\cdot \opt \geq (1-\eps)^3\cdot \opt$ with the first two properties stated in the lemma, and for every scenario $\kappa$ such that $k'<\kappa< k$ the profit is at least zero.

\paragraph{Partitioning the input into two independent problems.} Next, we use the above guessing step in order to partition the input into two independent inputs.  First, a job $j\in \J$ is called huge if its size is strictly larger than $UB$, that is, if $p_j> UB$ (and otherwise it is non-huge).

We pack every huge job into its own bag and we let $h$ denote the number of huge jobs (where we will show that $h \leq m-1$ holds).  Every huge job can cover one machine in every scenario $\kappa\geq k$ regardless of its exact size (because for all these scenarios we consider solutions for which $\opt_{\kappa}$ is not larger than $\rho\cdot LB=UB$).  On the other hand, the non-huge jobs (i.e., jobs of sizes at most $UB$) will be packed into bags of size at most $3\cdot UB$.  The packing of the non-huge jobs into bags will be optimized according to the scenarios of indexes at least $k$ (and at most $k_{\max}$), and we will ignore the scenarios of indexes smaller than $k$ when we pack the non-huge jobs into bags.  The resulting bags of the non-huge jobs together with the set of the huge jobs (which are bags consisting of a single job) will be packed into $\kappa\leq k'$ machines (in the corresponding scenario $\kappa$) based on existing EPTAS for the Santa Claus problem on identical machines (seeing bags as jobs).  We will show that if indeed we use bags of sizes at most $3\cdot UB$ when we pack the non-huge jobs into bags, then the scenarios of indexes at most $k'$ can be ignored. We show that this holds for any collection of bags of this form, that is, where every huge job has its own bag and other bags have total sizes of jobs not exceeding $3\cdot UB$. Thus, even though the bags are defined based on scenarios where the number of machines is at least $k$, still the scenarios with at most $k'$ machines have good solutions. We will also show that it is possible to restrict ourselves to these types of collections of bags.

\begin{lemma}\label{abc}
If all guesses are correct, any partition into bags has at least one bag with a total size of jobs no larger than $UB$. In particular, there are at most $m-1$ huge jobs.
\end{lemma}
\begin{proof}
Assume by contradiction that there is a partition into bags where every bag has a total size above $UB$.
 For every scenario $\kappa$ such that $1\leq \kappa \leq m$, it is possible to assign the bags such that the objective function value is larger than $UB$. Thus, for each scenario $\kappa$, it holds that $\opt_{\kappa} > UB \geq \frac{LB}{\eps^2}$, so $\opt>UB$, and thus $\opt > \frac{LB}{\eps^2}$. However, using the definition of $LB$ and the last bound, we have $LB \geq \frac{\eps \cdot \opt}{1+\eps}>{\eps/2 \cdot \opt}> \frac{LB}{2\eps}> LB$, a contradiction.
If there  are at least $m$ huge jobs, it is possible to create bags such that every bag has at least one huge job, which is a partition into bags where every bag has a total size above $UB$. Thus, the number of huge jobs does not exceed $m-1$.
\end{proof}

\begin{lemma}
An EPTAS for $\kappa$ machines where $\kappa \leq k'$ which is applied on bags (instead of jobs) such that every huge job has its own bag and any additional bag has total size of jobs not exceeding $3\cdot UB$ finds a solution of profit at least $\frac{\opt_{\kappa}}{(1-3\eps)\cdot(1-\eps)}$.
\end{lemma}
\begin{proof}
Consider an optimal solution for $\kappa$ machines and the original jobs. We show that this schedule can be modified into a schedule of the bags, such that the profit of the schedule is smaller by a factor not exceeding $1-3\eps$. Thus, an optimal schedule of the bags is not worse, and by using an EPTAS the profit may decrease by another factor of $1-\eps$.
The adaptation of the schedule is as follows. The huge jobs are assigned as before, since each of them has its own bag. All non-huge jobs are removed, and each machine receives bags until it is not possible to add another bag without exceeding the previous load or no unassigned bags remain. Given the property that the total size of bags is equal to the total size of jobs, it is either the case that the loads are equal to previous loads (in which case the value is unchanged) or there is at least one unassigned bag. In the latter case, no machine load decreased by more than an additive term of $3\cdot UB$, and therefore the value is at least the previous one minus $3\cdot UB$. To complete the assignment, all remaining bags are assigned arbitrarily. Since $\opt_{\kappa}\geq \frac{\rho\cdot LB}{\eps}=\frac{UB}{\eps}$ and the resulting value is at least $\opt_{\kappa}-3\cdot UB$, we find by $3\cdot UB \leq 3\eps \cdot \opt_{\kappa}$ that $\opt_{\kappa}-3\cdot UB \geq (1-3\eps)\cdot \opt_{\kappa}$.
\end{proof}

\begin{lemma}
Consider the instance of our problem when we let $q_{\kappa}=0$ for all $\kappa < k$ and for $\kappa > k_{\max}$. There exists a partition into bags where a  profit at least $\opt_\kappa$ can be obtained for any $\kappa\in[k,k_{\max}]$, the set of bags satisfies that every huge job is packed into its own bag, and each other bag consists of non-huge jobs of total size at most $2\cdot UB$.
\end{lemma}
\begin{proof}
The number of partitions into bags is finite for fixed $m,n$ (it does not exceed $m^n$). Consider the set of partitions for which a solution of profit at least $\opt_\kappa$ can be obtained for any $\kappa\in[k,k_{\max}]$. There is at least one such partition for the correct guess.
For every partition it is possible to define a vector of $m$ components, such that total sizes of bags appear in a non-increasing order. Consider the partition among the considered partitions for which the number of huge jobs that have their own bags is maximum, and out such partitions, one where the vector is lexicographically minimal. We claim that this partition satisfies the requirements.

Assume by contradiction that the number of huge jobs that do not have their own bags is not $h$. Consider a bag with a huge job that contains at least one additional job. This will be named the first bag.
Consider a bag whose total size is at most $UB$, which must exist due to Lemma \ref{abc}. This last bag does not have a huge job since the size of a huge job is above $UB$, and we call it the second bag.
Move all contents of the first bag into the second bag excluding one huge job. For any $\kappa\in[k,k_{\max}]$, since $\opt_{\kappa} \leq UB$, the assignment of bags to machines does not reduce the objective function value below $\opt_{\kappa}$. This holds because there is at most one machine whose total size was reduced (for every $\kappa\in[k,k_{\max}]$), but if such a machine exists, it still has a huge job whose size is above $UB$. Thus, the new partition is also one of the considered partitions.
The new partition has a larger number of huge jobs assigned into their own bags, because the second bag was not such a bag and the first bag became such a bag. This contradicts the choice of assignment to bags. Thus, the partition into bags consists of $h$ bags with huge jobs and $m-h \geq 1$ bags with non-huge jobs.

Next, consider only the components of the vector which do not correspond to bags of huge jobs. This vector is also minimal lexicographically out of vectors for partitions of the non-huge jobs into $m-h$ bags. Assume by contradiction that the first component is above $2\cdot UB$. Consider the bag of the first component (called the first bag now) and a bag with a total size at most $UB$ (called the second bag now). Move one job from the first bag to the second bag. The second bag now has a total size of at most $2\cdot UB$ (since a non-huge job was moved). The first bag now has a smaller total size, but still larger than $UB$. The sorted vector is now smaller lexicographically, and it is still possible to obtain a solution of profit at least $\opt_{\kappa}$ for any $\kappa\in[k,k_{\max}]$, similarly to the proof given here for huge jobs, which is a contradiction.
\end{proof}

In summary, we can focus on the scenarios interval $[k,k_{\max}]$, assume that huge jobs have their own bags, and remaining bags have total sizes not exceeding $2\cdot UB$.

\paragraph{Modifying the input of scenarios with indexes at least $\boldsymbol{k}$ and at most $\boldsymbol{k_{\max}}$.}
Motivated by the last partitioning of the original input into four parts, and the fact that only scenarios with indexes at least $k$ and at most $k_{\max}$ need to be considered, we apply the following transformation.
First, for every $\kappa <k$ we let $q_{\kappa}=0$, in the sense we can augment every solution to the remaining instance (with only a subset of scenarios) into an EPTAS for the original instance before this change to $q$. Furthermore, for every $\kappa>k_{\max}$ we let $q_{\kappa}=0$ in the sense we can ignore the profit of such scenarios.  Then, every huge job is packed into a separate bag.  Such a bag suffices to cover one machine in every remaining scenario.  Therefore, our second step in the transformation is the following one.  We delete the huge jobs from $\J$, we decrease the index of each remaining scenario by $h$ (in particular the new indexes in the support of $q$ will be in the interval $[k-h,k_{\max}-h]$), and we enforce the condition that every bag size is at most $2\cdot UB$.

As one can see, the transformations here are more complicated compared to those used in the previous section, and they have to be carefully designed and analyzed. However, now we are ready to apply methods that resemble the previous section.
Using the fact that for every remaining scenario we have that $\opt_{\kappa}$ is between $LB$ and $UB$, and the fact that $\frac{UB}{LB} = \rho \leq (\frac{1}{\eps})^{1/\eps+1}$ we are able to apply the methods we have developed for the makespan minimization problem to obtain an EPTAS for \prb, as we do next.

\subsection{Rounding bag sizes and job sizes}

\paragraph{Rounding bag sizes.}
We next provide a method to gain structure on the set of feasible solutions that we need to consider for finding a near optimal solution.  Given a feasible solution, the size of bag $i$ denoted as $P(i)$ is the total size of jobs assigned to this bag, that is, $P(i)=\sum_{j: \sigma_1(j)=i} p_j$.

\begin{lemma}\label{bigbag-santa}
There exists a solution of value at least $(1-3\eps)\cdot \opt$ in which the size of every non-empty bag is in the interval $[\eps \cdot LB, 2.5\cdot UB]$.
\end{lemma}
\begin{proof}
We note that by earlier arguments, we have that it suffices to consider a feasible solution subject to $P(i) \leq 2\cdot UB$ for all $i$.  Let $S$ be the set of bags in an optimal solution (whose value is  $\opt$) of sizes smaller than $\eps\cdot LB$.  In order to prove the claim we first process the optimal solution by applying the following bag merging process.  As long as there exists a pair of bags (of indexes) $i$ and $i'$ of sizes at most $\eps\cdot \opt$, we unite such a pair of bags.

After applying this change with respect to all such pairs we apply the following modifications to the schedules of the different scenarios.
For a given scenario $\kappa$, we apply the following process on every machine $i\leq \kappa$.  We compute the total size $s(i)$ of the bags in $S$ assigned to $i$ (both the assignment as well as the size of bags are with respect to the original optimal solution). We modify $\sigma_2^{(\kappa)}$ with respect to the bags in $S$. First, all these bags are removed from the machines (for this scenario). Then, the bags are assigned back to the machines. Machine $i$ will get a subset of these bags of total size at least $s(i)-2\eps \cdot LB$ and at most $s(i)$. Bags are assigned such that one bag is added to the machine at each time, and this is done until the next bag in the list of bags in $S$ does not fit the upper bound of $s(i)$ on the total size. Since the resulting size of any bag created by uniting the bags from $S$ is at most $2\eps \cdot LB$, we conclude that the total size of bags (from $S$) that are assigned to $i$ is at least $s(i) - 2\eps \cdot LB$ unless all machines satisfy these lower bounds. In this case we assign the remaining bags as follows.
For each remaining bag, select a machine $i$ that receives a total size of bags below $s(i)$ and add the bag to the machine. Due to the total size of bags, there is always such a machine as long as not all bags are assigned. Thus, for each machine the total size was increased by at most $2\eps \cdot LB$. The profit for scenario $\kappa$ increased by at most $2\eps \cdot LB$.
Thus, all bags in $S$ are assigned again.  The bags not in $S$ are assigned exactly as they used to be assigned in $\sigma_2^{(\kappa)}$, so the resulting value in the scenario is decreased by at most $2\eps \cdot LB$ with respect to $\opt_{\kappa}$.

After applying the merging process as long as possible, there might be at most one bag of size smaller than $\eps\cdot LB$, while all other bags have sizes of at least $\eps \cdot LB$. If such a bag exists, its jobs are reassigned to an arbitrary non-empty bag.  This last step may decrease the size of one bag by at most $\eps \cdot LB$ and this decreases the value of the solution of each scenario by at most $\eps \cdot LB$, so this bounds also the decrease of the expected value of the profit of the resulting solution after the second stage (using the fact that the sum of the components of $q$ is at most $1$). It may also increase the size of one bag by an additive term of at most $\eps\cdot LB$, and thus no bag has size below $\eps\cdot LB$, and no bag has size above $2\cdot UB+\eps\cdot LB \leq 2.5\cdot UB$.
\end{proof}

In what follows, we will decrease the size of a bag to be the next value of the form $(\eps + r\eps^2)\cdot LB$ for an integer $r \geq 0$ (except for empty bags for which the allowed size remains zero). Thus, we will not use precise sizes for bags but lower bounds on sizes, and we call them ``allowed sizes'' in what follows. For bags of allowed size zero the meaning remains that such a bag is empty.
We let $P'(i)$ be this decreased (allowed) size of bag $i$, that is, $P'(i) = \max_{r: (\eps + r\eps^2) \cdot LB \leq P(i)} (\eps + r\eps^2) \cdot LB$.  The allowed size is not smaller than $\eps \cdot LB$ and it is smaller than the size by an additive term of at most $\eps^2\cdot LB$. Thus, for every subset of bags, the total allowed size of the bags in the set is at least $\frac {1}{1+\eps}$ times the total size of the bags in the subset, we conclude that this rounding of the bag sizes will decrease the value of any solution by a multiplicative factor of at most $1+\eps$ (and therefore the expected value also decreases by at most this factor), and it will not increase the value of a solution for any scenario.
Thus, in what follows, we will consider only bag sizes that belong to the set $B=\{(\eps + r\eps^2) \cdot LB : r\in \Z, r\geq 0, (\eps + r\eps^2) \cdot LB \leq 2.5\cdot UB \}\cup\{0\}$.  We use this set by Lemma \ref{bigbag-santa}.
We conclude that the following holds.

\begin{corollary}
If the guessed information vector is satisfied by an optimal solution, then there is a solution of expected value of the Santa Claus value  of at least $\frac{1-3\eps}{1+\eps} \cdot \opt$ that uses only bags with allowed sizes in $B$.
\end{corollary}

Note that the allowed size of a bag may be slightly smaller than the actual total size of jobs assigned to the bag. Later, the allowed size acts as a lower bound (without an upper bound), and we sometimes take the allowed size into account in the calculation of machine completion times.

\paragraph{Rounding job sizes.}
We apply a similar rounding method for job sizes.  Recall that by our transformation, every job $j$ has size at most $UB$ (since huge jobs were removed from the input).  Next, we apply the following modification to the jobs of sizes at most $\eps^2 \cdot LB$.  While there is a pair of jobs of sizes at most $\eps^2 \cdot LB$, we unite such a pair of jobs. If there is an additional job of size at most $\eps^2 \cdot LB$, we delete it from the instance, and in the resulting solution (after applying the algorithm below on the resulting instance) we add the deleted job to an arbitrary bag.  This deletion of the deleted job decreases the Santa Claus value of every scenario by at most $\eps^2 \cdot LB$, so the resulting expected value of the Santa Claus value will be decreased by at most $\eps^2\cdot LB$. Adding the job back does not decrease the value. We consider the instance without this possibly deleted job, and we prove the following.
\begin{lemma}
The optimal expected value of the Santa Claus value of the instance resulting from the above modification of the job sizes among all solutions with allowed bag sizes in the following modified set
$B'=\{(\eps + r\eps^2) \cdot LB : r\in \Z, r\geq -2, (\eps + r\eps^2) \cdot LB \leq 3 \cdot UB \}\cup\{0\}$
is at least $(1-2\eps) \cdot \frac{1-3\eps}{1+\eps} \cdot \opt$.
\end{lemma}
\begin{proof}
We apply the following process on the solution $\sigma_1$ of the first stage.  For every bag $i$, we calculate the total size of jobs of size at most $\eps^2 \cdot LB$ assigned to $i$, and let $s_i$ be this value.  We delete the assignment of the jobs of size at most $\eps^2 \cdot LB$ from $\sigma_1$, we start defining the assignment of the new jobs that we created one by one in the following way.  We compute a lower bound $s_i - 2\eps^2 \cdot LB$ on the total size of these new jobs that we require to assign to bag $i$.  For every new job $j$, we assign it to a bag subject to the constraint that prior to this assignment, the new total size of new jobs in the bag will not exceed the lower bound. At the end every bag size is decreased by at most $2\eps^2 \cdot LB$ but perhaps there are remaining jobs. This is so because whenever a bag $i$ cannot get a new job, the previously assigned new jobs to $i$ have total size at most $s_i$.  By summing over all $i$, and noting that before the modifications we started with a feasible solution, we conclude that after considering all jobs all bags meet these lower bounds.  The remaining new jobs are assigned to bags for which the total size is smaller than it was, that is, each time a bag $i$ is selected such that it received less than a total size of $s_i$, and one remaining new job is added to that bag. As long as there is an unassigned job, there is also such a bag.
 In this process all bag sizes were decreased by additive terms of at most $2\eps^2\cdot LB$ and hence the modification of the allowed bag sizes. On the other hand, the increase of each bag size is at most $2\eps^2\cdot LB < \frac{UB}2$, because after bag $i$ exceeds the total size $s(i)$, the bag will not get another remaining job. Since only non-empty bags were affected, and the allowed size of a bag was at least $\eps \cdot LB$, the allowed size of a bag increases by at most a factor of $1+2\eps$.
 Thus, the new sizes do not exceed $3 \cdot UB$.
The claim holds by noting that the total allowed size of each bag was decreased by a multiplicative factor of at most $1-2\eps$ since previously all allowed bag sizes were at least $\eps \cdot LB$.
\end{proof}

Next, we round down the size of every job $j$ to be the next value that is of the form $(1+\eps)^r$ for an integer $r$. The size of every job cannot increase and it may decrease by a multiplicative factor of at most $1+\eps$. Job sizes are still larger than $\eps^3 \cdot LB$ and not larger than $UB \leq (\frac 1{\eps})^{1/\eps+1}\cdot LB$. Thus, the number of different job sizes is $O(\log_{1+\eps} (\frac {1}{\eps})^{1/\eps +4}) \leq \frac{2}{\eps^3}-1$.
We decrease the allowed size of each bag by another multiplicative factor of $1+\eps$. Bags sizes are rounded down again to the next value of the form $\eps + r\eps^2$, and since the smallest allowed size was $(\eps - 2\eps^2) \cdot LB$ and $\frac{\eps - 2\eps^2}{1+\eps} \geq \eps - 3\eps^2$, the allowed bag sizes become $B''=\{(\eps + r\eps^2) \cdot LB : r\in \Z, r\geq -3, (\eps + r\eps^2) \cdot LB \leq 3 \cdot UB \}\cup\{0\}$, and the expected value of the Santa Claus value of a feasible solution decreases by a multiplicative factor of at most $\frac{1-2\eps}{1-3\eps}$.
By our guessing step and our transformation, there is a feasible solution with expected value of the Santa Claus value of at least $ (1-\eps) (1-3\eps)^2 \cdot \opt$.

\medskip

A bag is called {\it tight} if the set of jobs of this bag cannot use a bag of a larger allowed size.  Note that any solution where some bags are not tight can be converted into one where all bags are tight without decreasing the expected value of the objective.

\begin{lemma}\label{PandP2}
For any tight bag it holds that the allowed size of the bag is at least $\eps^2$ times its actual size (the total size of jobs assigned to the bag, such that their rounded sizes are considered).
\end{lemma}
\begin{proof}
For a tight bag with allowed size at least $(\eps-3\eps^2)\cdot LB$ the claim is obvious since the total size is smaller than the allowed size plus $\eps^2\cdot LB$, and $\frac{\eps-3\eps^2}{\eps-2\eps^2}>\eps^2$.
 For tight bags of allowed size zero the property holds as well.
\end{proof}

\paragraph{Guessing an approximated histogram of the optimal Santa Claus value in all scenarios.}
For the remaining scenarios we use allowed bag sizes for computing profits of solutions. In order to find an upper bound on the modified values $\opt_{\kappa}$ (of the optimal solution for the original input) it is required to consider only steps where solutions were modified such that some bags became more loaded due to reassignment of jobs and some machines became more loaded due to reassignment of bags. Before any modifications, it holds that $\opt_{\kappa}\leq UB$. Since the profit increased by at most an additive term of $2\eps \cdot LB$ and then a multiplicative term of at most $1+2\eps$. We have $(1+2\eps)\cdot(UB +2\eps \cdot LB)<2\cdot UB$, and we use the upper bound $\opt_{\kappa}<2\cdot UB$. On the other hand, the modifications do not decrease any value $\opt_{\kappa}$ below $\frac{LB}2$.  That is, there exists a near optimal solution for which in every scenario $\kappa$, the value of the second stage solution is in $[\frac{LB}{2},2\cdot UB )$.
We also bound that number of bags assigned to a machine. We can assume that the starting time of a bag for scenario $\kappa$ is at most $\opt_{\kappa}$ since otherwise the bag can be moved to start at an earlier time. For each machine, at most one bag can start at time $opt_{\kappa}$. The number of bags that start earlier is below $\frac{2\cdot UB}{\eps\cdot LB}$, since the size of a bag is at most $\eps \cdot LB$. This number is at most $\frac{2\rho}{\eps}$ and together with at most one additional bag, the number is at most $\frac{3\rho}{\eps}$.

Our next guessing step is motivated by linear grouping of the $\opt_{\kappa}$ values for $\kappa\in [k,k_{\max}]$.  Once again, we consider a histogram. Here, the histogram corresponds to the values of the fixed optimal solution $\opt$ of $\prb$ for which the four guessed values correspond.
The width of the bar of scenario $\kappa$ is $q_{\kappa}$ and its height is the Santa Claus value in this scenario (that is, the value of $\opt_{\kappa}$).  By an earlier assumption,  when $\kappa$ is increased the optimal Santa Claus value does not increase. Thus, by plotting the bars in increasing order of $\kappa$ we get a monotone non-decreasing histogram, where the total area of the bars is the expected value of the Santa Claus values of the optimal solution and the total width of all the bars is denoted as $Q$ and it is at most $1$.
We sort the indexes of scenarios with (strictly) positive $q_{\kappa}$ in the interval $[k,k_{\max}]$ and we denote by $K$ the resulting sorted list.  For every $\kappa\in K$ we compute the total width of the bars with indexes at most $\kappa$ (using the ordering in $K$, i.e., in increasing indexes and considering only the indexes of the scenarios), and denote this sum by $Q_{\kappa}$ and let $Q'_{\kappa}=Q_{\kappa}-q_{\kappa}$.  Thus, $Q_{\kappa}$ is the total probability of scenarios in $\{k, k+1, \ldots,\kappa \}$, and $Q'_{\kappa}$ is the total probability of scenarios in $\{k, k+1,\ldots,{\kappa}-1\}$. The set $K$ does not contain scenarios with probability zero, but if $q_{{\kappa}-1},q_{\kappa}>0$, it holds that $Q_{{\kappa}-1}=Q'_{\kappa}$. According to the definitions, it holds that $Q_{\kappa}-Q'_{\kappa}=q_{\kappa}$, and the interval $[Q'_{\kappa},Q_{\kappa})$ on the horizontal axis corresponds to scenario ${\kappa}$.  Furthermore, $Q=Q_{k_{\max}}$.

Next, we compute a lower bound histogram $W$ as follows first by selecting a sublist $K'$ of $K$.  An index $\kappa \in K$ belongs to $K'$ if there exists an integer $\ell$ such that $Q'_{{\kappa}}\leq \frac{\eps^2}{\rho} \cdot \ell \cdot Q< Q_{{\kappa}}$. The motivation is that we would like to consider points which are integral multiples of $\frac{\eps^2}{\rho} \cdot Q$ and the set $K'$ contains all values of ${\kappa} \in K$ for which such special values belong to the interval $[Q'_{\kappa},Q_{\kappa})$. Values of Santa Claus value will be decreased such that the Santa Claus value function will still be piecewise constant, but it will have a smaller number of image values.

The new lower bound histogram is defined based on $K'$. For $\kappa \in K'$,  the value $\opt_{{\kappa}'}$ is used for all points with horizontal value between $Q_{\kappa}$ and up to $Q_{{\kappa}'}$ where ${\kappa}'>{\kappa}$ is the index just after ${\kappa}$ in the sublist $K'$, and if $\kappa$ is the largest element of $K'$, the value will be zero up to the last point where the histogram is defined (that is, up to $Q$).  Since the original histogram was monotone non-increasing, the new histogram is pointwise not larger than the original histogram. The possible modification is in the interval $[Q_{\kappa},Q'_{{\kappa}'})$ if ${\kappa}$ is not the maximum value is $K'$, and there may be a change for ${\kappa}+1,\ldots,{\kappa}'-1$ (that is, if ${\kappa}'\geq {\kappa}+2$). If ${\kappa}$ is the largest element of $K'$ but not of $K$, there may be a change for all ${\kappa}'\in K$ such that ${\kappa}'\geq {\kappa}+1$.
Thus, a lower bound on the total area below the histogram $W$ is not larger than the expected value of the Santa Claus value of the optimal solution.

We use the modified histogram $W$ to obtain another bound. For that we see $W$ as a step function whose values are the corresponding points in the upper edge of the histogram.  We define a new histogram $W'$ by letting it be $W(x-\frac{\eps^2}{\rho} \cdot Q)$ for all $x\in [0,Q\cdot(1+\frac{\eps^2}{\rho})]$ (and when $W$ is undefined due to a negative argument, we define it to be equal to $2\cdot UB$).  In this way we add a segment of length $\frac{\eps^2}{\rho}\cdot Q$ to $W$ and we shift the resulting histogram to the right to obtain $W'$. Since the Santa Claus value for every scenario is at least $\frac{LB}2$, we have that $\frac{LB}{2} \cdot Q$ is a lower bound on the area of the first histogram. On the other hand,  we added a segment of width of $\frac{\eps^2}{\rho} \cdot Q$ and height $2\cdot UB$, so the area of the segment is  $2\eps^2\cdot Q \cdot LB$.  The resulting histogram $W'$ is pointwise at least the original histogram (and thus also not smaller than $W$), and thus its area is also at least $\frac{LB}{2} \cdot Q$. For the added segment, this holds since the Santa Claus value for every scenario never exceeds $2\cdot UB = 2\cdot \rho\cdot LB$.
The relation between the three histograms holds also since every value $\opt_{\kappa'}$ was extended to the left for an interval not longer than $\frac{\eps^2}{\rho} \cdot Q$. Thus, if we consider $W$ instead of the original histogram, we decrease the value of the solution by a multiplicative factor smaller than $1+\eps$.

The guessing step that we use is motivated by this function $W$.  We let $W_{\kappa}$ be the height of the histogram $W$ in the bar corresponding to the scenario ${\kappa}$.  The relevant values ${\kappa}$ are those in $K'$, and the histogram $W$ only has values of the form $\opt_{\kappa}$ for ${\kappa} \in K'$. We guess the histogram $W$.  This means to guess the optimal Santa Claus value of $O(\frac{\rho}{\eps^2})$ scenarios.
Each of these has at most $\frac{2\rho}{\eps^2}$ different values since all allowed bag sizes are integer multiples of $\eps^2\cdot LB$ and the Santa Claus value of each scenario is at most $2\cdot UB = 2\rho \cdot LB$.  Thus, the number of possibilities of this guessing step is upper bounded by a function of $\frac{1}{\eps}$ which is $O({{(\frac{\rho}{\eps^2})}^{O(\frac{\rho}{\eps^2})}})$.
In what follows, we consider the iteration of the algorithm below when we use the value of the guess corresponding to the optimal solution.

\paragraph{The template-configuration integer program.}  A {\em template} of a bag is a multiset of job sizes assigned to a common bag.  We consider only multisets for which the total size of jobs is at most $3 \cdot UB $ due to the upper bound on allowed sizes of  bags.  Since the number of distinct job sizes (in the rounded instance) is upper bounded by $\frac{2}{\eps^3}-1$ and there are at most $\frac{3\rho}{\eps^2}$ jobs assigned to a common bag (since the rounded size of each job is above $\eps^2\cdot LB$), we conclude that the number of templates is at most $(\frac{2}{\eps^3 })^{(3\rho/\eps^2)}$, which is a constant depending only on $\frac{1}{\eps}$.
The reason is that each bag has $\frac{3\rho}{\eps^2}$ slots for jobs, such that each one may be empty or contain a job of one of the sizes.  We are going to have a (non-negative) counter decision  variable $y_t$ for every template $t$ that stands for the number of bags with template $t$.  Let $\tau$ be the set of templates, and assume that a template is a vector whose $\ell$-th component is the number of jobs with size equal to the $\ell$-th size which are packed into a bag with this template.  In order for such an assignment of the first stage to be feasible we have the following constraints where $n_{\ell}$ denotes the number of jobs in the rounded instance whose size is the $\ell$-th size of this instance. Recall that $m$ stands for the number of bags after the removal of some bags due to the packing of huge jobs.

$$ \sum_{t\in \tau} y_t \leq m \ , $$
$$ \sum_{t\in \tau} t_{\ell} \cdot y_t = n_{\ell}  \  , \forall \ell .$$

Observe that the number of constraints in the last family of constraints is a constant depending on $\eps$ (it is at most $\frac{2}{\eps^3}$), and all coefficients in the constraint matrix are (non-negative) integers not larger than $\frac{3\rho}{\eps^2}$.

We augment $K'$ with the maximum index in $K$ if this index does not already belong to $K'$, in order to satisfy the condition that for every index of $K$ there is some index of $K'$ that is not smaller.
Consider a scenario $\kappa\in K'$.  Recall that in scenario ${\kappa}$  the number of non-empty bags assigned to each machine is at most  $\frac{3\rho}{\eps}$. Define a configuration of a machine in scenario $\kappa$ to be a multiset of templates such that the multiset has at most $\frac{3\rho}{\eps}$ templates (taking copies into account) and the total size of the templates encoded in the multiset is at least $W_{\kappa}$ (the guess of a value of the histogram $W$). The number of configurations is at most a constant that depends only on the value of $\eps$, and this is also an upper bound on the number of suitable configurations. Since our mathematical program will not limit the Santa Claus value of any scenario using linear inequalities, we use a lower bound for it by not allowing configurations whose total allowed sizes is smaller than the planned Santa Claus value for the scenario. We let $C^{({\kappa})}$ denote the set of configurations for scenario ${\kappa}$, where $c\in C^{({\kappa})}$ is a vector of $|\tau|$ components where component $c_t$ for $t\in \tau$ is the number of copies of template $t$ assigned to configuration $c$.  Components are non-negative integers not larger than $\frac{3\rho}{\eps}$.
For scenario ${\kappa}\in K'$, we will have a family of (non-negative) decision variables $x_{c,{\kappa}}$ (for all $c\in C^{({\kappa})}$) counting the number of machines with this configuration.

For each such scenario ${\kappa}\in K'$, we have a set of constraints each of which involves only the template counters (acting as a family of global decision variables) and the family of the ${\kappa}$-th local family of decision variables, namely the ones corresponding to this scenario.  The family of constraints for the scenario ${\kappa}\in K'$ are as follows.

$$ \sum_{c\in C^{({\kappa})}} x_{c,{\kappa}}= {\kappa} \ ,$$
$$ \sum_{c\in C^{({\kappa})}} c_t\cdot x_{c,{\kappa}} - y_t = 0 \ , \ \forall t\in \tau \ .$$

Observe that the number of constraints in such a family of constraints is upper bounded by a constant depending only on $\eps$ (which is the number of possible templates plus $1$) and the coefficients in the constraint matrix are again all integers of absolute value at most a constant depending only on $\eps$ (at most $\frac{3\rho}{\eps}$.).

All decision variables are forced to be non-negative integers and we would like to solve the feasibility integer program defined above (with all constraints and decision variables).
The right hand side vector consists of integers that are at most $n$ (using $m <n$).
Such an integer linear program is again solved by Lenstra's algorithm whose time complexity is upper bounded (in our case) by some computable function of $\frac{1}{\eps}$ times a polynomial in $n$.

We apply such an algorithm and obtain a feasible solution $(x,y)$ if such a solution exists.  Based on such a feasible solution, we assign jobs to bags using the $y$ variables.  That is, for every $t\in \tau$, we schedule $y_t$ bags using template $t$.  By the constraint $ \sum_{t\in \tau} y_t \leq m$ there are at most $m$ bags, and the other bags will be empty.  Next, for every bag and every size $p$ of jobs in the rounded instance, if the template assigned to the bag has $\alpha$ jobs of this size, we will assign $\alpha$ jobs of this size to this bag.  Doing this for all sizes and all bags is possible by the constraints $ \sum_{t\in \tau} t_{\ell} \cdot y_t = n_{\ell} \  , \forall \ell$, and these constraints ensure that all jobs are assigned to bags. Note that the modification for jobs consisted of merging small jobs. When such a merged job is assigned, this means that a subset of jobs is assigned.  Next consider the assignment of bags to machines in each scenario.  Consider a scenario $\kappa' \in K$, and let $\kappa$ be the smallest index in $K'$ such that  $\kappa\geq \kappa'$. Assign $x_{c,\kappa}$ machines to configuration $c$ (for all $c\in C^{(\kappa)}$).  It is possible by the constraint  $ \sum_{c\in C^{({\kappa})}} x_{c,{\kappa}}= {\kappa}\geq \kappa'$.  If a configuration $c$ assigned to machine $i$ is supposed to pack $c_t$ copies of template $t$, we pick a subset of $c_t$ bags whose assigned template is $t$ and assign these bags to machine $i$.  We do this for all templates and all machines.  In this way we assign all bags by the constraints $ \sum_{c\in C^{({\kappa})}} c_t\cdot x_{c,{\kappa}} - y_t = 0  \ , \ \forall t\in \tau$.  In the case $\kappa'<\kappa$, bags that are supposed to be assigned to non-existent machines are assigned arbitrarily.
Since the configurations we used in scenario ${\kappa}$ have a total size of jobs of at least $W_{\kappa}$, we conclude the following.

\begin{corollary}
For every value of the guessed information for which the integer program has a feasible solution, there is a linear time algorithm that transform the feasible solution to the integer program into a solution to the rounded instance of \prb\ of Santa Claus value  at least $\sum_{q} q_k \cdot W_k$.
\end{corollary}

Our scheme is established by noting that an optimal solution satisfying the assumptions of the guessed information provides us a multiset of templates all of which are considered in $\tau$ and a multiset of configurations (and all of them have total size of templates not smaller than $W$) for which the corresponding counters satisfy the constraints of the integer program.  Thus, we conclude our second main result.

\begin{theorem}
Problem \prb\ admits an EPTAS.
\end{theorem}

\section{An EPTAS for $\boldsymbol{\lp}$ norm minimization\label{sec:eptas3}}
In Section \ref{sec:eptas1} we designed an EPTAS for makespan minimization. The main obstacle for reducing the running time and getting an EPTAS rather than a PTAS was that the guess of $\opt$ does not imply an upper bound on the total size of jobs assigned to a machine in some scenarios. We applied two new ideas in order to overcome it. The first one was elimination of scenarios with very small numbers of machines. The reason for doing this is that the makespan becomes large for any set of bags on one hand, but it is easy to obtain a good schedule for any reasonable collection of jobs on the other hand. The second idea is a reduction of the number of different makespans (or upper bounds on the makespan). The reason for applying this reduction  is that an integer program needs to be run for every collection of makespans, which would make the number of runs too large if scenarios with small probabilities are not combined to have the same upper bound on the makespan. For the $\lp$ norm, even knowing the cost of a scenario does not give us sufficient information on machine loads, but from this cost it is possible to add additional linear constraints to integer programs (which was not possible for our integer programs of the previous sections). The rounding steps are still possible but they are harder for $\lp$ norms (compared to the makepsan objective), and we start with the needed adaptations.

Here, the notation $\opt$ is for the smallest objective function value, and $\opt_{\kappa}$ is the objective function value for the best solution with this fixed set of bags and scenario ${\kappa}$. We assume without loss of generality that $\opt_{\kappa}\leq \opt_{\kappa'}$ for all ${\kappa}>\kappa'$,
The value of $\opt$ can be guessed up to a multiplicative factor of $1+\eps$ as before, due to the following lemma.

\begin{lemma}\label{boundsslp}
It holds that $\opt_{\kappa} \in [p_{\max},n \cdot p_{\max}]$. Thus, $\opt \in [p_{\max},n \cdot p_{\max}]$.
There is a set consisting of $O(\log_{1+\eps} n)$ values such that this set of values has at least one value in the interval $[\opt, (1+\eps)\opt)$.
\end{lemma}
\begin{proof}
It holds that $\opt_{\kappa} \geq p_{\max}$ since the $\lp$ norm is never smaller than the makespan and by Lemma \ref{boundss}. The upper bound holds  using a solution that assigns all bags to a common machine.  The bounds on $\opt$ hold due to linearity.
\end{proof}

Once again we let $\opt$ be the value of this guessed candidate value.  We will show that if there is a feasible solution to \prc\ whose expected value of the cost is at most $\opt$, then we will construct a feasible solution with expected value of the cost being at most $(1+\Theta(\eps))\cdot \opt$.

\paragraph{First rounding step of the $\boldsymbol{\opt_{\kappa}}$ values.} Our first rounding step has a similar role to the use of $LB$ in the EPTAS for \prb.  Here, our first rounding step is to round up the values of $\opt_{\kappa}$ for all ${\kappa}$ to the next integer multiple of $\eps\cdot \opt$.  That is, we will treat $\opt_{\kappa}$ as an upper bound of the $\lp$ norm of the second stage solution in scenario ${\kappa}$ and this upper bound is rounded up in this rounding step. The rounding increases each cost by an additive term below $\eps\cdot \opt$.

\begin{lemma}
The rounding up of the values of $\opt_{\kappa}$ to the next integer multiples of $\eps\cdot \opt$ increases the cost of the solution to \prc\ (that is, the value of $\sum_{\kappa} q_{\kappa}\cdot\opt_{\kappa}$) by at most $\eps\cdot\opt$.
\end{lemma}
\begin{proof}
The claim follows as $\sum_{\kappa} q_{\kappa} =1$.
\end{proof}

Thus, in what follows we assume that there is a feasible solution to the problem whose cost is at most $(1+\eps)\cdot \opt$ such that for every scenario ${\kappa}$ the value of the $\lp$ norm of the solution in scenario ${\kappa}$ is at most the rounded value of $\opt_{\kappa}$. We mostly use this rounding for $k_{\max}$ but in order to keep monotonicity all values $\opt_{\kappa}$ are rounded here.

For every ${\kappa}$, with a slight abuse of notation, we let $\opt_{\kappa}$ be the rounded up value of the cost of this solution in scenario ${\kappa}$.
We let $k_{\max}$ be the maximum value in the support of $q$ (note that the definition is different from the previous section) and observe that by the last rounding step we have that $\opt_{k_{\max}} \geq \eps\opt$.  Since $\opt_{k_{\max}} $ is the smallest element in a weighted average that is at most $(1+\eps)\opt$, we conclude that $\opt_{k_{\max}} \leq (1+\eps)\opt$.

\paragraph{The required splitting of the instance.}
Similarly to \prb\ we would like to split the instance into three parts (where some of the parts may be empty). There will also be a prefix (which could be empty) of scenarios for which the specific assignment into bags is carried out by assigning each relatively large job to its own bag while other bags are very small compared to the load of each machine in an optimal solution for the scenario. For the remaining scenarios, the prefix will consist of scenarios for which the total cost is small and thus could be inflated by a large factor, and a suffix of the most important scenarios.

Let $LB$ be equal to $\frac{\opt_{k_{\max}}}{k_{\max}^{1/\pp}}$ and by our initial rounding step there are at most $O(\frac{1}{\eps})$ values that it may attain for a fixed value of $\opt$.  The guessing of $\opt_{k_{\max}}$ within a multiplicative factor of $1+\eps$ suffices for the guessing of $LB$ within a multiplicative factor of $1+\eps$.  For a scenario $\kappa$, we let $\iota(\kappa)$ denote the {\it target load} of the machines in the scenario that is $\iota(\kappa) = \frac{\opt_{\kappa}}{\kappa^{1/\pp}}$.  While this value may be irrational, we do not compute it explicitly.
The motivation of this definition is that a vector with $\kappa$ components each of which equals $\iota(\kappa)$ has an $\lp$ norm of exactly $\opt_{\kappa}$.  In particular $LB=\iota(k_{\max})$.  Furthermore, since by our assumption $\opt_k\leq \opt_{k'}$ if $k>k'$, we conclude that $\iota(k)\leq \iota(k')$ if $k>k'$ (the numerators are a monotonically non-increasing function while the denominators are monotonically increasing).
The proof of the next lemma is obtained by selecting a set of consecutive scenarios minimizing the total cost out of $\frac{1}{\eps^2}$ such disjoint sequences of scenarios, and requires the knowledge not only of $\opt$ (which we guessed) but of all values of the form $\iota(i)$. However, we use the lemma to obtain the information that given the correct value $LB$, there is a pair of values $k,k'$ and a value $\rho$ that specify the only properties of $\opt$ (except for the value) that we use for our algorithm.

\begin{lemma}\label{yylp}
There is a value of $\rho$ which is an integer power of $\frac{1}{\eps}$ that satisfies $\frac{1}{\eps^2} \leq \rho \leq (\frac{1}{\eps})^{1/\eps^2+1}$ and there exist two indexes $k,k'$ such that $0 \leq k'<k \leq k_{\max}$, where every non-zero index has to be a scenario in the support of $q$, such that the following conditions hold.
\begin{enumerate}
\item $\iota(k) \leq \rho \cdot LB$,
\item If $k'\geq 1$, then $\iota(k') \geq \frac{\rho}{\eps} \cdot LB$,
\item $\sum_{\kappa=k'+1}^{k-1} q_{\kappa}\cdot \opt_{\kappa} \leq \eps^2\cdot(1+\eps) \cdot \opt$.
\end{enumerate}
\end{lemma}
\begin{proof}
For every value of $\rho=\frac{1}{\eps^r}$ for $r=2,3,\ldots ,\frac{1}{\eps^2}+1$, we compute  the total weighted cost $L_r$ of scenarios $\kappa$ for which the value of the target load for this scenario in the optimal solution is within the interval $[LB \cdot \frac{1}{\eps^r},LB \cdot \frac{1}{\eps^{r+1}})$. That is, we let $L_r=\sum_{\kappa:\iota(\kappa)\in[LB \cdot \frac{1}{\eps^r},LB \cdot \frac{1}{\eps^{r+1}})}q_{\kappa}\cdot \opt_{\kappa}$. For this definition we use the monotonicity of $\iota(\kappa)$ as a function of $\kappa$ (it holds that $\iota(\kappa_1)>\iota(\kappa_2)$ for $\kappa_1<\kappa_2$), so we consider consecutive intervals of scenarios again. Note that the sums are parts of the objective function value while the summation is for values of $\iota(\kappa)$.
A minimizer $r'$ is selected, and let $\rho=\frac{1}{\eps^{r'}}$.

Let $k$ be the minimum index of a scenario in the support of $q$ for which $\iota(k) \leq \rho \cdot LB$ and $k'$ be the maximum index of a scenario for which $\iota(k') \geq \frac{\rho}{\eps} \cdot LB$. Since $\iota(k_{\max}) = LB < \rho\cdot LB$, the integer $k$ is well-defined. If $k'$ does not exist in the sense that there is no scenario $k''$ for which $\iota(k'') \geq \frac{\rho}{\eps} \cdot LB$, we let $k'=0$. The upper bound on the partial weighted sum (values of the form $q_{\kappa}\cdot \opt_{\kappa}$) is due to the number of options for $r$ and the rounding which may increase $\opt$ by a factor of at most $1+\eps$, so the sum of all values  is at most $(1+\eps)\cdot \opt$. This proves the third condition.
\end{proof}

\paragraph{The next guessing step.} In this step we guess several additional values. We guess the value of $LB$ (by guessing $\opt_{k_{\max}}$ after rounding, as explained above), the value of $\rho$ (as an integer power of $\frac{1}{\eps}$), and two indexes $k'<k\leq k_{\max}$.  That is, we would like to enumerate a set of polynomial number of candidate values of four components vectors that contains at least one vector with first component value of exactly $LB$, with second component  value is exactly $\rho$, and the two indexes $k'<k$ in the third and fourth components of that vector.  We will show that for the guess that corresponds to the optimal solution, we indeed find a solution of expected value of the $\lp$ norm of at most  $(1+\Theta(\eps))\cdot \opt$. The next lemma shows that it is possible to enumerate such a set of candidate values. Our algorithm performs this enumeration.

\begin{lemma}
There is a set consisting of $O(\frac{m^2}{\eps^3})$ vectors such that this set of vectors contains at least one vector with the required properties. Thus, the number of guesses is at most $O(\frac{m^2}{\eps^3})$.
\end{lemma}
\begin{proof}
The bound of $m$ on the number of possibilities for the values of $k$ and of $k'$ is based on the number of scenarios. The number of options for $\rho$ is $\frac{1}{\eps^2}$, and the number of options for the  value of $LB$ is $O(\frac{1}{\eps})$.
\end{proof}

In what follows, we  let $LB$ be the first component value of the guessed information, $\rho$ be the second component value of the guessed information, and $k,k'$ be the two guessed scenarios (where it is possible that $k'=0$ is not an actual scenario).

We let $UB=LB \cdot \rho$. Furthermore, for every scenario $\kappa$ with $k'<\kappa< k$, we increase $\iota(\kappa)$ to be $\frac{UB}{\eps}$ (where previously it was in the interval $[\rho\cdot LB,\rho\cdot UB)$). This inflates $\opt_{\kappa}$ to be a larger value, and the new value of $\opt_{\kappa}$ for such a scenario will be $\kappa^{1/\pp} \cdot \frac{UB}{\eps}$. It is possible that as a result those last values are no longer monotone.

We use Lemma \ref{yylp} to find the next properties. Since $\iota(\kappa) < LB \cdot \frac{1}{\eps^{r'+1}} = \frac{UB}{\eps}$, this step indeed increases $\iota(\kappa)$ and thus also increases $\opt_{\kappa}$. Since $\iota(\kappa) \geq LB \cdot \frac{1}{\eps^{r'}} = UB$, $\iota(\kappa)$ was increased by a factor of at most $\frac{1}{\eps}$ and this is also the bound on the increase of $\opt_{\kappa}$. The increase of the $\opt_{\kappa}$ values in this interval increases the weighted value of the cost $\sum_{\kappa=k'+1}^{k-1} q_{\kappa}\cdot\opt_{\kappa}$ by at most $\frac{1}{\eps}$ times its previous value. Let the previous value be $\lambda$. The new value is at most $\frac 1{\eps}\cdot \lambda$, and the difference is at most $(\frac 1{\eps}-1)\cdot \lambda \leq \frac{1-\eps}{\eps}\cdot \eps^2\cdot(1+\eps) \cdot \opt \leq \eps\cdot \opt$. In this way we combined scenarios $k'+1,\ldots,k-1$ (if $k>k'+1$) with scenarios $1,2,\ldots,k'$, and created a gap between $\iota(k)$ and $\iota(k-1)$.

\paragraph{Partitioning the input into two independent problems.} We use the above guessing step in order to partition the input into two independent inputs.

We consider  $UB$ as a threshold value, such that an imaginary schedule for $k$ machines for which every machine has load $UB$, will have a cost of at least $\opt_k$. Thus, a schedule for $\kappa \geq k$ machines, where every machine has a load of at least $UB$ and at least one machine has a strictly larger load, cannot have a cost of at most $\opt_k$.  In particular, for any set of bags, any schedule for $\kappa \geq k$ machines has at least one machine of load at most $UB$. Moreover, the total size of jobs does not exceed $k\cdot UB$, since otherwise we find that even in an optimal fractional schedule for $k$ machines, where each machine has the same load, the cost will exceed $\opt_k$.

Next, a job $j\in \J$ is called huge if its size is strictly larger than $2\cdot UB$, that is, if $p_j> 2\cdot UB$ (and otherwise it is non-huge).  We let $\J_{h}$ and $\J_{nh}$ denote the set of huge jobs and the set of non-huge jobs, respectively. This partition into sets is based on the guessed values and in particular on $LB$, $\rho$, and $k$.

We let $h$ denote the number of huge jobs, where we intend to pack every huge job into its own bag. The motivation is that there will be a sufficiently small number of such jobs, and the bag of every huge job is assigned alone to some machine in every scenario $\kappa\geq k$ regardless of its exact size (following the analysis of Alon et al. \cite{alon1997approximation}, and based on the fact that for all these scenarios we consider solutions for which $\iota(\kappa)$ is not larger than $UB$).  On the other hand, the non-huge jobs (i.e., jobs of sizes at most $2 \cdot UB$) will be packed into bags of size at most $7\cdot UB$.  The packing of the non-huge jobs into bags will be optimized according to the scenarios of indexes at least $k$, and we will ignore the scenarios of indexes smaller than $k$ when we pack the non-huge jobs into bags.  The resulting bags of the non-huge jobs together with the set of the huge jobs (which are bags consisting of a single job) will be packed into $\kappa < k$ machines (in the corresponding scenario) based on existing EPTAS for the $\lp$ norm minimization problem on identical machines (seeing bags as jobs).  We will show that indeed we are going to use bags of sizes at most $7\cdot UB$ when we assign the non-huge jobs to bags, and the scenarios of indexes smaller than $k$ can be ignored in what follows. This holds for any collection of bags of this form, that is, where every huge job has its own bag, and other bags have total sizes of jobs not exceeding $7\cdot UB$. Thus, even though the bags are defined based on scenarios where the number of machines is at least $k$, still the scenarios with less than $k$ machines have good solutions.

\begin{lemma}\label{abclp}
If all guesses are correct, any partition into bags has at least one bag with a total size of jobs no larger than ${UB}$. In particular, there are at most $k-1$ huge jobs.
\end{lemma}
\begin{proof}
Assume by contradiction that there is a partition into bags where there are at least $k$ bags with total sizes above $UB$, or there are at least $k$ huge jobs. This means that the total size of jobs exceeds $k\cdot {UB}$, which gives a cost exceeding $\opt_k$ in scenario $k$ and thus $\iota(k)>UB$.
\end{proof}

\begin{lemma}
An EPTAS for $\kappa$ machines where $\kappa <k$ (after inflating the $\opt_{\kappa}$ values for $k'<\kappa <k$) which is applied on bags (instead of jobs) such that every huge job has its own bag and any additional bag has total size of jobs not exceeding $7\cdot UB$ finds a solution of $\lp$ norm value of at most $\opt_{\kappa}\cdot (1+\eps)(1+7\eps)$.
\end{lemma}
\begin{proof}
Consider an optimal solution for $\kappa$ machines and the original jobs (without using specific bags, that is, all jobs are assigned directly to machines), where the cost is at most $\opt_{\kappa}$ (which is possible because the last cost is the outcome of increasing the cost to that of an optimal solution for a fixed partition to bags). We show that this schedule can be modified into a schedule of a fixed set of bags whose $\lp$ norm value is larger by a multiplicative factor not exceeding $1+7\eps$. An optimal schedule of the bags cannot be worse than this schedule, and by using an EPTAS for the bags rather than computing an optimal schedule the resulting cost may increase by another factor of $1+\eps$.

The adaptation of the schedule is as follows. The huge jobs are assigned as before, since each of them has its own bag. All non-huge jobs are removed, and each machine receives bags until it is not possible to add another bag without exceeding the previous load plus an additive term of $7\eps \cdot UB$,  or no unassigned bags remain. Given the property that the total size of bags is equal to the total size of jobs, if there is an unassigned bag, this means that at least one machine now has a smaller load compared to its previous load.  Since no bag has a total size above $7\eps \cdot UB$, an unassigned bag can be assigned without exceeding the original load of the machine that receives it by more than $7\eps \cdot UB$.

Since assigning fractions of jobs to each machine among the $\kappa$ machines such that all machines have equal loads results in a work of at least $UB$ on every machine (using $\kappa<k$), we conclude that the vector of dimension $\kappa$ that is $UB$ in every component, has an $\lp$ norm of at most $\opt_{\kappa}$. So the cost of the resulting schedule is not larger than $(1+7\eps) \opt_{\kappa}$ by the triangle inequality of the norms.
\end{proof}

\begin{lemma}
Consider the instance of our problem when we let $q_{\kappa}=0$ for all $\kappa < k$. There exists a partition into bags where an $\lp$ norm of at most $\opt_\kappa$ can be obtained for any $\kappa \geq k$, and the set of bags satisfies that every huge job is packed into its own bag, and any other bag consists of non-huge jobs of total size at most $4\cdot UB$.
\end{lemma}
\begin{proof}
The number of partitions into bags is finite for fixed $m,n$. Consider the set of partitions for which a solution of cost at most $\opt_\kappa$ can be obtained for any $\kappa \geq k$. There is at least one such partition for the correct guess.
For every partition it is possible to define a vector of $m$ components where the $i$-th component is the total size of jobs in the $i$-th bag, such that total sizes of bags appear in a non-increasing order. Consider the partition out of the considered partitions for which the number of huge jobs that have their own bags is maximum, and out of such partitions, one where the vector is lexicographically minimal. We claim that this partition satisfies the requirements.

Given a set of bags, for each scenario $\kappa\geq k$ we consider a solution that minimizes the cost (where this cost does not exceed $\opt_{\kappa}$), and assigns all empty bags (bags with total sizes of zero) to the least loaded machine (breaking ties arbitrarily). We show that this solution has the following two properties. The first property is that any bag of size above $UB$ is scheduled alone to a machine. The second property is that no machine has a load above $2\cdot UB$ unless it has a bag of size above $UB$. Recall that since $\kappa\geq k$, it holds that  $\opt_{\kappa} \leq UB \cdot k^{1/\pp}$, and therefore at least one machine has load of at most $UB$. Thus, all empty bags are assigned to a machine of load at most $UB$.

Assume that the first property does not hold. Consider a bag of size above $UB$ which we call the first bag that is assigned with at least one other bag, where any such other bag has a positive size. We call this machine and this bag the first machine and the first bag.
Consider a solution where a minimum size bag from the first machine is moved to a machine of load at most $UB$. We call this machine and this bag the second machine and the second bag.
Let $z_1$ denote the load of the first machine excluding the second bag. Let $z_2>0$ be the size of the second bag. Let $z_3$ be the load of the second machine excluding the second bag. The loads of the two machines were $z_1+z_2$ and $z_3$ and now they are $z_1$ and $z_2+z_3$. It holds that $z_1>UB$ and $z_3\leq {UB}$ which implies $z_3<z_1$. Thus by using also $z_2>0$, we have $z_3<z_1<z_1+z_2$ and $z_3<z_2+z_3<z_1+z_2$. By strict convexity, the sum of the two $\pp$ powers of loads decreases.
Thus, the assignment of bags did not minimize the cost, a contradiction.

Assume that the second property does not hold. This means that there is a machine with a load above $2\cdot {UB}$ and only bags of sizes at most ${UB}$. Additionally, there is again a machine with load at most ${UB}$ (and its bags also have sizes of at most ${UB}$). Thus, all bags of the first machine have positive sizes.
Moving a minimum bag from the first considered machine to the second considered machine is similar to the previous proof, but here $z_1>{UB}$ since $z_1+z_2> 2\cdot {UB}$ while $z_2\leq {UB}$.

Next, assume by contradiction that the number of huge jobs that do not have their own bags is not $h$, and consider a bag that contains a huge job with at least one other job, and we call it the first bag. By the first property above, this bag is the only bag of its machine in all examined solutions. Consider a bag whose total size is at most $2 \cdot {UB}$, which must exist due to Lemma \ref{abclp}, and call it the second bag.
The second bag does not have a huge job since the size of a huge job is above $2 \cdot {UB}$. Due to the second property, the second bag is assigned to a machine with load at most $2 \cdot {UB}$ for every scenario $\kappa\geq k$.
Move all contents of the first bag into the second bag excluding one huge job.
Consider a scenario $\kappa \geq k$. Let the sizes of the huge job, the jobs previously sharing the first bag with it, and the load of the machine of the second bag in scenario $\kappa$ be $z_1$, $z_2$, and $z_3$, respectively.
Before the modification of bags, the loads of the two machines for scenarios $\kappa$ were $z_1+z_2$ and $z_3$, and after the modification of bags they are $z_1$ and $z_2+z_3$. It holds that $z_1 > 2 \cdot {UB}$ and $z_3 \leq 2 \cdot {UB}$. Thus, $z_3 < z_1 \leq z_1+z_2$ and $z_3 \leq z_2+z_3 < z_1+z_2$, and the modification of bags does not increase the objective function value above $\opt_{\kappa}$ by the convexity of the $\pp$ power. This contradicts the choice of assignment to bags. Thus, the partition into bags consists of $h$ bags with huge jobs and $m-h \geq 1$ bags with non-huge jobs.

Finally, consider only the components of the vector which do not correspond to bags of huge jobs. This vector is also minimal lexicographically out of vectors for partitions of the non-huge jobs into $m-h$ bags. Assume by contradiction that the first component is above $4\cdot {UB}$. Consider the bag of the first component which we call the first bag and a bag with a total size at most ${UB}$ which we call the second bag. The first bag has only jobs of sizes at most $2 \cdot {UB}$.
Move one job from the first bag to the second bag. We use the same notation as in the previous proof, and here $z_1+z_2>4 \cdot {UB}$, $z_2\leq 2\cdot {UB}$, and $z_3\leq {UB}$. Thus, $z_3<z_2+z_3 \leq 4\cdot {UB}<z_1+z_2$ and $z_3< z_1\leq z_1+z_2$.
\end{proof}

In summary, we can focus on the scenarios interval $[k,k_{\max}]$, assume that huge jobs have their own bags, and remaining bags have total sizes not exceeding $4\cdot {UB}$.

\paragraph{Modifying the input of scenarios with indexes at least $\boldsymbol{k}$.}
Motivated by the last partitioning of the original input into two parts, and the fact that only scenarios with indexes at least $k$ need to be considered, we apply the following transformation.
First, for every $\kappa <k$ we let $q_{\kappa}=0$, in the sense we can augment every solution to the remaining instance into an EPTAS for the original instance before this change to $q$.  Then, every huge job is packed into a separate bag.  Such a bag will be scheduled to its own machine in every remaining scenario.  Therefore, our second step in the transformation is the following one.  We delete the huge jobs from $\J$,  we let the new $\opt_{\kappa}$ value be $\opt'_{\kappa}=\left( \opt_{\kappa}^{\pp} - \sum_{j \in \J_{h}} p_j^{\pp}  \right)^{1/\pp}$  we decrease the index of each remaining scenario by $h$ (in particular the new indexes in the support of $q$ will be in the interval $[k-h,k_{\max}-h]$), and we enforce the condition that every bag size is at most $4\cdot {UB}$.  Note that we have not modified the values of $LB$ and $UB$. With a slight abuse of notation we will use $k$ (and $\kappa$) instead of $k-h$ (and $\kappa-h$, respectively). That is, we assume without loss of generality that there are no huge jobs.

Using the fact that for every remaining scenario we have that $\iota({\kappa})$ is between $LB$ and $UB$, and the fact that $\frac{UB}{LB} = \rho \leq (\frac{1}{\eps})^{1/\eps^2+1}$, we will be able to apply the methods we have developed for the makespan minimization problem to obtain an EPTAS for \prc, as we do next.

We consider the rounding of the bag sizes and job sizes similarly to our rounding in Section \ref{sec:eptas1}.

\paragraph{Rounding bag sizes.}
We provide a method to gain structure on the set of feasible solutions that we need to consider for finding a near optimal solution.  Bag sizes are denoted again by $P(i)$ for bag $i$.

\begin{lemma}\label{bigbaglp}
There exists a solution of cost at most $(1+3\eps)\cdot \opt$ in which the size of every bag is in the interval $[\eps \cdot {LB}, 5 \cdot {UB}]$.
\end{lemma}
\begin{proof}
Let $S$ be the set of bags in an optimal solution with sizes of at most $\eps \cdot {LB}$.  In order to prove the claim we first process the optimal solution by applying the following bag merging process.  As long as there exists a pair of bags (of indexes) $i$ and $i'$, and sizes at most $\eps\cdot {LB}$, we unite such a pair of bags.

After applying this change with respect to all such pairs we apply the following modifications to $\sigma_2$.
For a given scenario $\kappa$ whose $\iota({\kappa})\geq LB$ (by the monotonicity of the target load values for the scenarios), we apply the following process on every machine $i\leq \kappa$.  We compute the total size $s(i)$ of the bags in $S$ assigned to $i$ (both the assignment as well as the size of bags are with respect to the original optimal solution). We modify $\sigma_2^{(\kappa)}$ with respect to the bags in $S$. First, all these bags are removed from the machines (for scenario $\kappa$). Then, the merged bags are assigned back to the machines. Machine $i$ will get a subset of these bags of total size at most $s(i)+2\eps \cdot {LB}$. Bags are assigned such that one bag is added to the machine at each time, and this is done until the next bag in the list of bags in $S$ does not fit the upper bound on the total size. Since the resulting size of any bag created by uniting the bags from $S$ is at most $2\eps \cdot {LB}$, we conclude that the total size of bags (from $S$) that are assigned to $i$ is at least $s(i)$ unless all bags were already assigned to earlier machines and there are no unassigned bags.  Thus, all bags in $S$ are assigned again.  The bags not in $S$ are assigned exactly as they used to be assigned in $\sigma_2^{(\kappa)}$, so the resulting work of each machine in scenario $\kappa$ is increased by at most $2\eps \cdot {LB}$ with respect to its value in the optimal solution of the second stage in scenario $\kappa$ (for the original solution of the first stage).

After applying this process as long as possible, there might be at most one bag of size smaller than $\eps\cdot {LB}$ and its jobs are reassigned to an arbitrary non-empty bag.  This last step may increase the size of one bag by at most $\eps \cdot {LB}$ and this increases the work of one machine of each scenario by at most $\eps \cdot {LB}$.
Together, for every scenario $\kappa$, the load of each machine is increased by at most $3\eps \cdot {LB}$. We separate a vector with components of $3\eps \cdot {LB}$ from the load vector, and use the triangle inequality for norms. The remaining vector has cost not above $\opt_{\kappa}$. The separated vector has an $\lp$ norm of at most the $\lp$ norm of a vector of dimension $k_{\max}$ whose components are all equal to $3\eps\cdot LB$ and this $\lp$ norm value is at most $3\eps \cdot \opt_{k_{\max}} \leq 3\eps\cdot \opt_{\kappa}$.
We find that this bounds also the increase of the expected value of the cost of the resulting solution after the second stage. The claim holds by linearity of the expected value.

After the modification process, every bag has size at least $\eps \cdot {LB}$ and sizes of bags were increased by at most an additive term of $3\eps \cdot {LB} $. Bag sizes were previously at most $4 \cdot {UB}$ so now they are at most $5 \cdot {UB}$.
\end{proof}

In what follows, we will increase the size of a bag to be the next value of the form $(\eps + r\eps^2)\cdot {LB}$ for a non-negative integer $r$ (except for empty bags for which the allowed size remains zero). Thus, we will not use precise sizes for bags but upper bounds on sizes, and we call them ``allowed sizes'' in what follows.
We let $P'(i)$ be this increased (allowed) size of bag $i$, that is, $P'(i) = \min_{r: (\eps + r\eps^2) \cdot {LB} \geq P(i)} (\eps + r\eps^2) \cdot {LB}$.  Since for every subset of bags, the total allowed size of the bags in the set is at most $1+\eps$ times the total size of the bags in the subset, we conclude that this rounding of the bag sizes will increase the cost of any solution by a multiplicative factor of at most $1+\eps$ (using the linearity of the $\lp$ norm, and therefore the expected value also increases by at most this factor).
Thus, in what follows, we will consider only bag sizes that belong to the set $B=\{(\eps + r\eps^2) \cdot {LB} : r\geq 0, (\eps + r\eps^2) \cdot {LB} \leq 5\cdot {UB} \}\cup\{0\}$.
We use this set by Lemma \ref{bigbaglp}.
We conclude that the following holds.

\begin{corollary}
If the guessed information is correct, then there is a solution of expected value of the $\lp$ norm of at most $(1+\eps)(1+3\eps) \cdot \opt$ that uses only bags with allowed sizes in $B$.
\end{corollary}

Note that the allowed size of a bag may be slightly larger than the actual total size of jobs assigned to the bag. Later, the allowed size acts as an upper bound (without a lower bound), and we take the allowed size into account in the calculation of machine completion times.

\paragraph{Rounding job sizes.}
We apply a similar rounding method for job sizes.  Recall that by the partitioning rule, every job $j$ is non-huge and thus has size at most $2\cdot {UB}$.  Next, we apply the following modification to the jobs of sizes at most $\eps^2 \cdot {LB}$.  Whenever there is a pair of jobs, each of which has size at most $\eps^2 \cdot {LB}$, we unite the two jobs.  This is equivalent to restricting our attention to solutions of the first stage where we add the constraint that the two (deleted) jobs must be assigned to a common bag.  We repeat this process as long as there is such a pair.  If there is an additional job of size at most $\eps^2 \cdot {LB}$, we delete it from the instance, and in the resulting solution (after applying the steps below on the resulting instance) in every scenario we add the deleted job to bag $1$.  This addition of the deleted job increases the $\lp$ norm of every scenario by at most $\eps^2 \opt_{k_{\max}}$ that is at most $\eps^2$ times the optimal cost of the corresponding scenario. So the resulting expected value of the $\lp$ norm will be increased by at most $\eps^2\opt$.  We consider the instance without this possible deleted job, and we prove the following.

\begin{lemma}
The optimal expected value of the \lp\ of the instance resulting from the above modification of the job sizes among all solutions with allowed bag sizes in the following modified set
$B' =\{(\eps + r\eps^2) \cdot {LB} : r\geq 0, (\eps + r\eps^2) \cdot {LB} \leq 6 \cdot {UB} \}\cup\{0\}$    is at most
$(1+2\eps) (1+\eps) (1+3\eps) \cdot \opt$.
\end{lemma}
\begin{proof}
We apply the following process on the solution $\sigma_1$ of the first stage.  For every non-empty bag $i$, we calculate the total size of jobs of size at most $\eps^2 \cdot {LB}$ assigned to $i$, and let $s_i$ be this value.  We delete the assignment of the jobs of size at most $\eps^2 \cdot {LB}$ from $\sigma_1$, we start defining the assignment of the new jobs that we created one by one in the following way.  We compute an upper bound $s_i+2\eps^2 \cdot {LB}$ on the total size of these new jobs that we allow to assign to bag $i$.  For every new job $j$, we assign it to a bag subject to the constraint that after this assignment, the new total size of new jobs in the bag will not exceed the upper bound.  This is possible, because whenever a bag $i$ cannot get a new job, the previously assigned new jobs to $i$ have total size at least $s_i$.  By summing over all $i$, and noting that before the modifications we started with a feasible solution, we conclude that there is always a feasible bag to use.  In this process all non-empty bag sizes were increased by an additive terms of $2\eps^2 \cdot {LB}$ which increases sizes by a factor of at most $1+2\eps$ since previously all allowed bag sizes were at least $\eps \cdot {LB}$, and hence the modification of the allowed bag sizes. The change in the cost of the solution holds due to the same reason as well as the increased upper bound on the bag sizes.
\end{proof}

We next round up the size of every job $j$ to be the next integer power of $1+\eps$. The size of every job cannot decrease and it may increase by a multiplicative factor of at most $1+\eps$. Job sizes are still larger $\eps^2 \cdot {LB}$ and at most $2(1+\eps)\cdot {UB}$. Thus, the number of different job sizes is $O(\log_{1+\eps} \frac {2\rho}{\eps^2}) \leq O(\frac{1}{\eps^4})$.
We increase the allowed size of each bag by another multiplicative factor of $1+\eps$ and round it up again to the next value of the form $(\eps + r\eps^2) \cdot {LB}$.
Thus, the expected value of the \lp\ norm of a feasible solution increases by a multiplicative factor of at most $(1+\eps)^2$ and the maximum bag size increases by a multiplicative factor of at most $(1+\eps)^2\leq \frac 76$.
So by our guessing step, there is a feasible solution with expected value of the \lp\ norm of at most $(1+\eps)^3(1+2\eps)(1+3\eps) \opt$ that uses only bags with allowed size in  $B'' = \{(\eps + r\eps^2)\cdot {LB} :  r\geq 0, (\eps + r\eps^2) \cdot {LB} \leq 7 \cdot {UB} \} \cup\{0\}$.

\medskip

A bag is called {\it tight} if the set of jobs of this bag cannot use a bag of a smaller allowed size.  Note that any solution where some bags are not tight can be converted into one where all bags are tight without increasing the cost. The proof of the next lemma is identical to the proof for \pra.

\begin{lemma}\label{PandP3}
For any tight bag it holds that the allowed size of the bag is at most $\frac{1}{\eps}$ times its actual size (the total size of jobs assigned to the bag, such that their rounded sizes are considered).
\end{lemma}

\paragraph{Guessing an approximated histogram of the optimal \lp\ norm values on all scenarios.}  Our next guessing step is motivated by linear grouping of the $\opt_{\kappa}$ values.  Consider a histogram corresponding to the costs of a fixed optimal solution satisfying all above structural claims, where the width of the bar corresponding to scenario $\kappa \geq k$ is $q_{\kappa}$ and its height is the \lp\ norm value in this scenario (that is, $\opt_{\kappa}$).  Observe that when $\kappa$ is increased, the optimal \lp\ norm value does not decrease, so by plotting the bars in an increasing order of $\kappa$ we get a monotone non-decreasing histogram, where the total area of the bars is the expected value of the \lp\ norm of the optimal solution and the total width of all the bars is at most $1$. We sort the indexes at least $k$ of scenarios with (strictly) positive $q_{\kappa}$ and we denote by $K$ the resulting sorted list.  For every $\kappa\in K$ we compute the total width of the bars with indexes at least $k$ and at most $\kappa$ (using the ordering in $K$), and denote this sum by $Q_{\kappa}$ and let $Q'_{\kappa}=Q_{\kappa}-q_{\kappa}$.  Observe that the minimum height of a point in the histogram is $\opt_{k_{\max}}$ and its maximum value is not larger than $2\opt_k$ (since in each scenario $\kappa$ we have perhaps increased the $\opt_{\kappa}$, but the rounding steps could increase the costs for every scenario by a small multiplicative factor not exceeding $2$).  Here, we have the following bound on their ratio $$\frac{2\opt_k}{\opt_{k_{\max}}} = \frac{2\iota(k) \cdot k^{1/\pp}}{\iota(k_{\max}) \cdot k_{\max}^{1/\pp}} \leq \frac{2\iota(k)}{\iota(k_{\max})} \leq 2 \rho ,$$ so the number of integer powers of $1+\eps$ in the interval between these values is a function of $1+\eps$.

Next, we compute an upper bound histogram $W$ as follows first by selecting a sublist $K'$ of $K$.  An index ${\kappa}\in K$ belongs to $K'$ if there exists an integer $\ell$ such that $Q'_{{\kappa}}\leq \frac{\eps^2}{\rho} \cdot \ell < Q_{{\kappa}}$. Values of \lp\ norm will be increased such that the \lp\ norm  function will still be piecewise linear, but it will have a smaller number of image values.

The new upper bound histogram is defined as $\opt_{\kappa}$ for all points with horizontal value between $Q'_{\kappa}$ and up to $Q'_{{\kappa}'}$ where ${\kappa}'>{\kappa}$ is the index just after ${\kappa}$ in the sublist $K'$ or up to the last point where the histogram is defined, $Q_{k_{\max}}$ if ${\kappa}$ is the largest element of $K'$.  Since the original histogram was monotone non-increasing, the new histogram is pointwise not smaller than the original histogram.

We use the modified histogram $W$ to obtain another bound. For that we see $W$ as a step function whose values are the corresponding points in the upper edge of the histogram.  We define a new histogram by letting it be $W(x+\frac{\eps^2}{\rho})$ for all $x\in [0,1]$.  In this way we delete a segment of length $\eps^2/\rho $ from $W$ and we shift the resulting histogram to the left. Since the \lp\ norm for every scenario never exceeds $2\cdot \opt_k$, every point in $W$ has a height of at most $2  \cdot \rho $  times its minimum height and we deleted a segment of width of $\eps^2/\rho $.  Thus, the total area we delete is at most $\eps \cdot \opt_{k_{\max}}$.  The resulting histogram is pointwise at most the original histogram (and thus also not larger than $W$).  This property holds since every value $\opt_{\kappa}$ was extended to the right for an interval not longer than $\eps^2/\rho$. Thus, if we consider $W$ instead of the original histogram, we increase the cost of the solution by a multiplicative factor not larger than $1+3\eps$.

The guessing step that we use is motivated by this function $W$.  We let $W_{\kappa}$ be the height of the histogram $W$ in the bar corresponding to the scenario $\kappa$ rounded up to the next integer power of $1+\eps$.  The relevant values ${\kappa}$ are those in $K'$, and the histogram $W$ only has values of the form $(1+\eps)^{\lceil \log_{1+\eps} \opt_{\kappa}\rceil}$ for ${\kappa} \in K'$ (this means that $\opt_{\kappa}$ was rounded up to a power of $1+\eps$ and multiplied by at most $1+\eps$). We guess the histogram $W$.  This means to guess an approximated value of the optimal \lp\ norm value of $O(\frac{\rho}{\eps^2})$ scenarios. As explained above, the number of possibilities for each value of an approximated $\lp$ norm  for a scenario in $K'$ is upper bounded by a function of $\frac{1}{\eps}$.  Thus, the number of possibilities of this guessing step is upper bounded by a function of $\frac{1}{\eps}$.
In what follows, we consider the iteration of the algorithm below when we use the value of the guess corresponding to the optimal solution.

\paragraph{The template-configuration integer program.}  A {\em template} of a bag is a multiset of job sizes assigned to a common bag.  We consider only multisets for which the total size of jobs is at most $7 \cdot {UB}$ since the largest allowed size of any bag is not larger.  Since the number of distinct job sizes (in the rounded instance) is upper bounded by $\frac{2\rho}{\eps^3}-1$ and there are at most $\frac{7\rho}{\eps^2}$ jobs assigned to a common bag (since the rounded size of each job is above $\eps^2\cdot {LB}$), we conclude that the number of templates is at most $(\frac{2\rho}{\eps^3})^{(7\rho/\eps^2)}$, which is a constant depending only on $\frac{1}{\eps}$.
We are going to have a non-negative counter decision variable $y_t$ for every template $t$ that stands for the number of bags with template $t$.  Let $\tau$ be the set of templates, and assume that a template is a vector whose $\ell$-th component is the number of jobs with size equal to the $\ell$-th size are packed into a bag with this template.  In order for such assignment of the first stage to be feasible we have the following constraints where $n_{\ell}$ denotes the number of jobs in the rounded instance whose size is the $\ell$-th size of the instance.

$$ \sum_{t\in \tau} y_t \leq m \ , $$
$$ \sum_{t\in \tau} t_{\ell} \cdot y_t = n_{\ell} \  , \forall \ell .$$

We augment $K'$ with the minimum index in $K$ if this index does not already belong to $K'$, similarly to our action for the makespan objective.
Consider a scenario $\kappa \in K'$.  Recall that in scenario $\kappa$  the number of non-empty bags assigned to each machine is at most  $\frac{4\rho}{\eps}$. Define a configuration of a machine in scenario $\kappa$ to be a multiset of templates such that the multiset has at most $\frac{4\rho}{\eps}$ templates (including copies of templates).  For such a configuration $c$, we let $s(c)$ denote the total size of the templates in the multiset encoded by $c$.

The number of configurations is at most the number of templates to the power of $\frac{4\rho}{\eps}$.
Here, unlike the earlier schemes we will not screen the configurations but we will use a linear inequality in the integer program to upper bound the total \pp\ power of the $s(c)$ values. This is done because there is no threshold load for machines on one hand, but the cost is based on a sum on the other hand.

 We let $C^{({\kappa})}$ denote the set of configurations for scenario $\kappa$, where $c\in C^{(\kappa)}$ is a vector of $|\tau|$ components where component $c_t$ for $t\in \tau$ is the number of copies of template $t$ assigned to configuration $c$.  Components are non-negative integers not larger than $\frac{4\rho}{\eps}$.
For scenario $\kappa\in K'$, we will have a family of non-negative decision variables $x_{c,\kappa}$ (for all $c\in C^{(\kappa)}$) counting the number of machines with this configuration.

For each such scenario $\kappa \in K'$, we have a set of constraints each of which involve only the template counters (acting as a family of global decision variables) and the family of the $\kappa$-th local family of decision variables, namely the ones corresponding to this scenario.  The family of constraints for the scenario $\kappa$ are as follows.

$$ \sum_{c\in C^{(\kappa)}} x_{c,\kappa}= \kappa \ ,$$
$$ \sum_{c\in C^{(\kappa)}} c_t\cdot x_{c,\kappa} - y_t = 0 \ , \ \forall t\in \tau \ .$$
$$ \sum_{c\in C^{(\kappa)}} (s(c))^{\pp} \cdot x_{c,\kappa} \leq (W_{\kappa})^{\pp} \ . $$
Observe that the number of constraints in such a family of constraints is upper bounded by a constant depending only on $\eps$ (which is the number of possible templates plus $2$).  Observe that for a given scenario there are only a constant number of decision variables corresponding to the scenario and since $K'$ has a constant cardinality, we have a constant number of decision variables.

All decision variables are forced to be non-negative integers and we would like to solve the feasibility integer program defined above (with all constraints and decision variables).
We use Lenstra's algorithm \cite{lenstra1983integer,Kan83} for integer programming in fixed dimension to solve the integer program we have exhibited.
We apply the algorithm and obtain a feasible solution $(x,y)$ if such a solution exists.
In order to solve the above integer program in strongly polynomial time, we note that all coefficients of the constraint matrix and the right hand side are bounded by a constant depending on $1/\eps$ or by a polynomial in $m+n$ except for the new constraints we have introduced for upper bounding the total \pp-power of the size of the configurations is at most the guessed value of the histogram to the power of \pp. In order to achieve the situation where the coefficients of these constraints are replaced by a (single-) exponential function in $m+n$, we can use a constructive reducibility bound due to Frank and Tardos \cite{frank1987application} (see also \cite{EHKKLO23}). This result allows to replace the original linear constraint with a different constraint where all coefficients are bounded by $2^{d^3} \cdot n^{d^2}$ so that an integer non-negative vector in $\R^d$ that is at most $n$ in every component, satisfies the old inequality if and only if it satisfies the new inequality. In our case, $d$ is a function of $\eps$. Then, using the new inequality instead of the old one (and doing that for every scenario $\kappa\in K'$) would result a feasibility integer program for which Lenstra's algorithm \cite{lenstra1983integer,Kan83} would run in strongly polynomial time in $n$ and $m$ (times a computable function of $\frac{1}{\eps}$).

Based on such a feasible solution, we assign jobs to bags using the $y$ variables.  That is, for every $t\in \tau$, we schedule $y_t$ bags using template $t$.  By the constraint $ \sum_{t\in \tau} y_t \leq m$ there are at most $m$ bags, and the other bags will be empty.  Next, for every bag and every size $p$ of jobs in the rounded instance, if the template assigned to the bag has $\alpha$ jobs of this size, we will assign $\alpha$ jobs to this bag.  Doing this for all sizes and all bags is possible by the constraints $ \sum_{t\in \tau} t_{\ell} \cdot y_t = n_{\ell} \  , \forall \ell$, and these constraints ensure that all jobs are assigned to bags.  Next consider the assignment of bags to machines in each scenario.  Consider a scenario $\kappa'\in K$.  We find the minimum value $\kappa\in K'$ such that $\kappa\leq \kappa'$.  We will exhibit a solution with only $\kappa$ identical machines. We assign $x_{c,\kappa}$ machines to configuration $c$ (for all $c\in C^{(\kappa)}$).  It is possible by the constraint  $ \sum_{c\in C^{(\kappa)}} x_{c,\kappa}= \kappa$.  If a configuration $c$ assigned to machine $i$ is supposed to pack $c_t$ copies of template $t$, we pick a subset of $c_t$ bags whose assigned template is $t$ and assign these bags to machine $i$.  We do this for all templates and all machines.  In this way we assign all bags by the constraints $ \sum_{c\in C^{(\kappa)}} c_t\cdot x_{c,\kappa} - y_t = 0 \ , \ \forall t\in \tau$. Note that the modification for jobs consisted of merging small jobs. When such a merged job is assigned, this means that a subset of jobs is assigned. Since the configurations we used in scenario $\kappa$ satisfies that $ \sum_{c\in C^{(\kappa)}} (s(c))^{\pp} \cdot x_{c,\kappa} \leq (W_{\kappa})^{\pp} $, we conclude that the resulting schedule in this scenario has an \lp\ norm of at most $W_{\kappa}$.  Therefore, we have established the following.

\begin{corollary}
For every value of the guessed information for which the integer program has a feasible solution, there is a linear time algorithm that transform the feasible solution to the integer program into a solution to the rounded instance of \prc\ of cost at most $\sum_{q} q_k \cdot W_k$.
\end{corollary}

Our scheme is established by noting that an optimal solution satisfying the assumptions of the guessed information provides us a multiset of templates all of which are considered in $\tau$ and a multiset of configurations  for which the corresponding counters satisfy the constraints of the integer program.  Thus, we conclude our third main result.

\begin{theorem}
Problem \prc\ admits an EPTAS.
\end{theorem}

\bibliographystyle{abbrv}

\end{document}